\documentclass[11pt,dvipsnames]{article}
\usepackage[utf8]{inputenc}

\usepackage{amsmath,amsfonts,amssymb,amsthm}
\usepackage{algorithmic}
\usepackage[ruled, vlined, linesnumbered,nofillcomment]{algorithm2e}
\usepackage[a4paper]{geometry}

\usepackage{array}
\usepackage[caption=false,font=normalsize,labelfont=sf,textfont=sf]{subfig}
\usepackage{textcomp}
\usepackage{stfloats}
\usepackage{url}
\usepackage{verbatim}
\usepackage{graphicx}
\usepackage{cite}
\usepackage{euscript}
\usepackage{bbm}

\usepackage{tikz}
\usetikzlibrary{patterns}

\usetikzlibrary{decorations.pathreplacing}
\usepackage{stmaryrd} 
\usepackage{url}

\usepackage[final,breaklinks=true,colorlinks=true,citecolor=RubineRed,linkcolor=blue,pdfusetitle=true]{hyperref} 

\newenvironment{algo}[1]{
  \algorithm[ht]
    \caption{#1}
    \DontPrintSemicolon
    \SetAlgoCaptionLayout{left}
    \SetAlgoHangIndent{0pt}
    \SetKwInOut{Input}{input}
    \SetKwInOut{Output}{output} 
}{
  \endalgorithm
}

\newtheorem{theo}{Theorem} 

\newtheorem{lemma}[theo]{Lemma}

\usepackage{thmtools}
\usepackage{thm-restate}

\newcommand{\error}{\boldsymbol{e}}

\def\cG{\mathcal{G}}

\def\1{\mathbbm{1}}

\def\eps{\varepsilon}

\DeclareUnicodeCharacter{00A0}{ }

\newcommand{\F}{\mathbb F}
\newcommand{\G}{\EuScript G}

\newcommand{\C}{\EuScript C}

\newcommand{\eQ}{\EuScript Q}
\newcommand{\Q}{Q}

\newcommand{\f}{\mathbb F}
\newcommand{\x}{\mathbf x}

\newcommand{\bari}{\overline{\imath}}
\newcommand{\barj}{\overline{\jmath}}

\renewcommand{\geq}{\geqslant}
\renewcommand{\leq}{\leqslant}

\newcommand{\Cay}{{\sf Cay}}

\newcommand{\AL}[1]{\textcolor{blue}{#1}}

\begin{document}

\title{Decoding quantum Tanner codes}

\author{Anthony Leverrier\thanks{Anthony Leverrier is with Inria Paris, France (e-mail: anthony.leverrier@inria.fr)},\
 Gilles Z{\'e}mor\thanks{Gilles Z\'emor is with the Institut de Math\'ematiques de Bordeaux UMR 5251,
Universit\'e de Bordeaux, 351 cours de la Lib\'eration - F33405, Talence, France (e-mail: zemor@math.u-bordeaux.fr).}}
\maketitle
\markboth{Transactions on Information Theory,~Vol., No., Month~year}%
{Leverrier, Z{\'e}mor: Decoding quantum Tanner codes}


\begin{abstract}
We introduce sequential and parallel decoders for quantum Tanner codes. When the Tanner code construction is applied to a sufficiently expanding square complex with robust local codes, we obtain a family of asymptotically good quantum low-density parity-check codes. In this case, our decoders provably correct arbitrary errors of weight linear in the code length, respectively in linear or logarithmic time. The same decoders are easily adapted to the expander lifted product codes of Panteleev and Kalachev. 
 Along the way, we exploit recently established bounds on the robustness of
random tensor codes to give a tighter bound on the minimum distance of quantum Tanner codes. 
\end{abstract}

%

\section{Introduction and Overview}

Quantum low-density parity-check (LDPC) codes hold the promise of drastically
reducing the required overhead for fault-tolerant quantum computing compared to
approaches based on the surface code \cite{got14,KP13,FGL18b,TDB21,EB21}.
Rather surprisingly, this advantage was already established for the class of
hypergraph product codes \cite{TZ14} that predates the recent series of
breakthroughs \cite{HHO21,BE21b,PK20} culminating with the
discovery of asymptotically good quantum LDPC codes \cite{PK21,LZ22}, that is, quantum codes of length $n$ with a constant encoding rate and a minimum distance $d$ linear in $n$. These new codes should significantly improve the performance of  quantum computers that can implement long-range gates between arbitrary qubits, for instance ion-based \cite{CK13}, photonic \cite{rud17}, or Rydberg-based \cite{MW21} architectures.
To take full advantage of these capabilities, it is crucial to devise efficient
decoding algorithms so as to keep the errors under control during an execution
of a quantum algorithm. This in particular requires highly parallelisable
decoders that run in logarithmic time, since new errors keep accumulating while the classical decoder tries to identify errors that have occurred earlier. 
We present such an algorithm for the family of quantum Tanner codes \cite{LZ22}. The same decoder also serves as the main subroutine for the decoding of the expander lifted product codes of Panteleev and Kalachev \cite{PK21}.

There are two interesting settings for decoding algorithms, depending on whether 
the errors are arbitrary (or adversarial) or whether they follow a simple stochastic model (such as independent and identically distributed errors or local stochastic errors \cite{got14}). Before the recent breakthroughs yielding codes with a linear minimum distance, the distinction was crucial because experimentally relevant errors have a weight linear in the code length (since each qubit suffers an error with constant probability) and could only be dealt with by making some assumptions about their distribution. This is problematic in the context of fault tolerance because correlations between errors are essentially impossible to track down. 
While several decoders perform reasonably well against random noise, even with a noisy syndrome extraction \cite{PK21b, GGK21, RWB20, QVR21, DLB22, BL22, CMS22}, they cannot handle much more than $\sqrt{n}$ errors in an adversarial setting \cite{LTZ15, EKZ20}.

We will first review the construction of quantum Tanner codes and then present
our logarithmic-time decoding algorithm. It is strongly inspired by a mostly
sequential, linear-time decoder analysed in \cite{LZ22b}, and is an adaptation of the small-set-flip decoder initially designed for hypergraph product codes \cite{LTZ15}, which was also recently applied to other good quantum LDPC codes \cite{GPT22,DHL22,LH22b}.
Along the way we will give a sharper estimate for the minimum distance of
quantum Tanner codes and present a sequential decoder that is significantly
simpler than that of \cite{LZ22b}, both conceptually and technically.

{\bf Quantum Tanner codes}.---
These codes are a generalisation of classical Tanner codes \cite{T81, SS96}.
They are obtained by enforcing local linear constraints (corresponding to the
dual of a small tensor code of constant size) of constant weight on qubits
associated with the 2-faces (squares) of a square complex \cite{DEL22,BE21b}.
The square complex that we use appears in the work of Panteleev and Kalachev~\cite{PK21} and can also be
thought of as a quadripartite version of the left-right Cayley complex of Dinur \textit{et al.}\ \cite{DEL22}. 
It is an incidence structure between a set $V$ of vertices,
two sets of edges $E_A$ and $E_B$, that we will refer to as $A$-edges and
$B$-edges, and a set $Q$ of squares. 
The vertex-set $V$ is partitioned into four subsets $V=V_{00} \cup V_{01} \cup
V_{10} \cup V_{11}$, corresponding to four copies of a fixed group $G$, that is, $V_{ij} = G \times \{i,j\}$. We also have two self-inverse subsets $A=A^{-1}$ and $B=B^{-1}$ of the group $G$ and assume for simplicity that $A$ and $B$ are of the same cardinality $\Delta$. For $i \in \{0,1\}$, two vertices $v=(g,i0)\in V_{i0}$ and
$v'=(g',i1) \in V_{i1}$ are related by an $A$-edge if $g'=ag$ for some $a\in A$.
Similarly, for $j \in \{0,1\}$, vertices $v = (g,0j)$ and $v' = (g',1j)$ are related by a $B$-edge if $g'=gb$ for
some $b\in B$. The sets $E_A$ and $E_B$ make up the set of $A$-edges and
$B$-edges respectively and define two graphs $\G_A = (V, E_A), \G_B = (V,E_B)$,
each of which consists of two
disjoint copies of the double cover of a Cayley graph over the group $G$ (with
generator set $A$ for $\G_A$, and $B$ for $\G_B$).
Next, the set $Q$ of squares is defined as the set of $4$-subsets of vertices of the form $\{(g,00),(ag,01),(agb,11),(gb,10)\}$, with the four vertices belonging to distinct copies of $G$.

If we restrict the vertex set to $V_0 := V_{00} \cup V_{11}$, every square is now
incident to only two vertices: one in $V_{00}$ and one in $V_{11}$. The set of squares can then be
seen as a set of edges on $V_0$, and it therefore defines a bipartite graph that we denote
by $\G_0^{\square}=(V_0,Q)$. Similarly, the restriction to the vertices of $V_1 := V_{01} \cup V_{10}$ defines the
graph $\G_1^{\square}$, which is an exact replica of $\G_0^{\square}$: both graphs are 
defined over two copies of the group $G$, with $g,g,'\in G$ being related by an edge
whenever $g'=agb$ for some $a\in A, b\in B$. 
For any vertex $v$, we denote by $Q(v)$ the $Q$-neighbourhood of $v$
defined as the set of squares incident to $v$. The $Q$-neighbourhood
$Q(v)$ has cardinality $\Delta^2$ and is isomorphic to the product set $A\times
B$: the situation is illustrated on Fig.~\ref{fig:code}.

\begin{figure}[h]
\begin{center}
\resizebox{8.8cm}{6cm}{
\begin{tikzpicture}
\draw (0,0) rectangle (3,3);
\draw[pattern=north west lines,pattern color = blue] (1,0) rectangle (1.5,3);
\draw[pattern=north east lines,pattern color = red] (0,1.5) rectangle (3,2);
\draw[step=0.5cm] (0,0) grid (3,3);

\draw (4,0) rectangle (7,3);
\draw[pattern=north west lines,pattern color = blue] (5,0) rectangle (5.5,3);
\draw[pattern=north east lines,pattern color = red] (5,1.5) rectangle (5.5,2);

\draw[step=0.5cm] (4,0) grid (7,3);

\draw (0,4) rectangle (3,7);
\draw[pattern=north east lines,pattern color = red] (0,5.5) rectangle (3,6);
\draw[pattern=north west lines,pattern color = blue] (1,5.5) rectangle (1.5,6);
\draw[step=0.5cm] (0,4) grid (3,7);

\draw (4,4) rectangle (7,7);
\draw[step=0.5cm] (4,4) grid (7,7);
\draw[pattern=north east lines,pattern color = red] (5,5.5) rectangle (5.5,6);
\draw[pattern=north west lines,pattern color = blue] (5,5.5) rectangle (5.5,6);

\node at (-1.,0.5) {$Q(g,00)$};
\node at (8.,0.5) {$Q(gb,10)$};
\node at (-1.,6.5) {$Q(ag,01)$};
\node at (8.,6.5) {$Q(agb,11)$};
\node at (1.25,-0.4) {$b$};
\node at (5.35,-0.35) {$b$};
\node at (-0.4,1.75) {$a$};
\node at (-0.4,5.75) {$a$};

\end{tikzpicture}
}
\end{center}
\caption{The four $Q$- neighbourhoods $Q(v)$ that contain the square $\{ (g,00), (ag, 01), (agb,11), (gb,10)\}$ depicted in red and blue. The $Q$- neighbourhoods of two vertices connected by an $A$-edge (resp.~a $B$-edge) share a row depicted in red (resp.~a column in blue). The labeling is chosen to ensure that a given square, such as the one in red and blue, is indexed similarly, by $(a,b)$ here, in the four $Q$- neighbourhoods. The $\sigma_x$-type generators are codewords of $C_A\otimes C_B$ in the $Q$- neighbourhoods of $V_{00} \cup V_{11}$; the $\sigma_z$-type generators are codewords of $C_A^\perp \otimes C_B^\perp$ in the $Q$- neighbourhoods of $V_{01} \cup V_{10}$. They automatically commute since their support can only intersect on a shared row or column (as depicted), and the orthogonality of the local codes ensure that they commute on this row or column.
\label{fig:code}}
\end{figure}
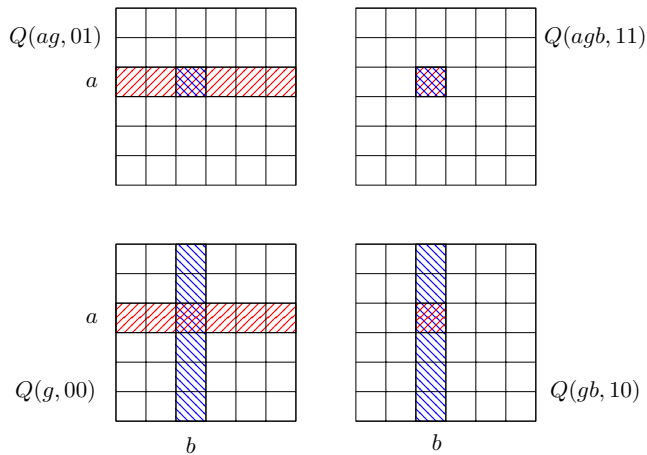

A classical \emph{Tanner code} on a $\Delta$-regular graph $\G=(V,E)$ is the
set of words of $\F_2^E$ such that
on the edge neighbourhood of every vertex $v\in V$, we see a codeword of a small code $C$ of length $\Delta$~\cite{T81}. We denote the resulting code by $\text{Tan}(\G, C) \subset \F_2^E$. Expander codes, obtained by combining good codes $C$ with an expander graph $G$, gave the first explicit family of good classical LDPC codes with a constant encoding rate, a linear minimum distance and an efficient decoder \cite{SS96}.

\emph{Quantum Tanner codes} are quantum CSS codes formed by two classical Tanner codes $\C_0$ and $\C_1$ with support on the set $Q$ of squares of a square complex as above. The CSS construction \cite{ste96b,CS96} requires both codes to satisfy the orthogonality condition $\C_0^\perp \subset \C_1$. To this end, we define local codes on the space $\f_2^{A\times B}$ that we may think of as the space of matrices whose rows (columns) are indexed by $A$ (by $B$).
If $C_A\subset \f_2^A$ and $C_B\subset\f_2^B$ are two linear codes,
we define the {\em tensor} (or product) code $C_A\otimes C_B$ as the space of
matrices $x$ such that for every $b\in B$ the column vector $(x_{ab})_{a\in A}$
belongs to $C_A$ and for every $a\in A$ the row vector $(x_{ab})_{b\in B}$
belongs to $C_B$.
Recalling that the dual $C^\perp$ of a code $C\subset\F_2^\Delta$ is the set of
words orthogonal to all words in $C$, we define $\C_0$ and $\C_1$ to be the classical Tanner codes
$\C_0 := \text{Tan}(\G_0^\square,(C_A \otimes C_B)^\perp)$ and $\C_1 := \text{Tan}(\G_1^\square,(C_A^\perp \otimes C_B^\perp)^\perp)$,
with bits associated to each square of $\Q$ and local constraints enforced at the vertices of $V_0$ and $V_1$, respectively. 
To check the orthogonality condition between the two codes, it is convenient to
look at their generators (or parity-checks). We define a $\sigma_x$-generator
for $\C_0$ (resp.~a $\sigma_z$-generator for $\C_1$) as a vector of $\F_2^\Q$
whose support lies entirely in the $\Q$- neighbourhood $\Q(v)$ of $V_{0}$
(resp.~$V_1$), and which is equal to a codeword of $C_A\otimes C_B$
(resp.~$C_A^\perp\otimes C_B^\perp$) on $\Q(v)$. The Tanner code $\C_0$ (resp.~$\C_1$) is defined as the set of vectors orthogonal to all $\sigma_x$-generators (resp.~$\sigma_z$-generators).
The commutation between both types of generators follows from the fact that if a
$\sigma_x$-generator on $v_0 \in V_0$ and a $\sigma_z$-generator on $v_1 \in V_1$ have
intersecting supports, then $v_0$ and $v_1$ must be neighbours in the left-right Cayley complex and their $Q$- neighbourhoods must intersect on either a column or a row, on which the two generators equal codewords of $C_A$ and $C_A^\perp$, or of $C_B$ and $C_B^\perp$ (see Fig.~\ref{fig:code}). Since $\Delta$ is chosen constant with respect to $n$, we see that the generators of the code have constant weight and that each qubit only appears in a constant number of generators: the resulting quantum Tanner code is therefore a quantum LDPC code by definition. Choosing $C_A$ and $C_B$ of rates $\rho$ and $1-\rho$, so that both $C_0$ and $C_1$ have rate $\rho(1-\rho)$, yields a quantum code with encoding rate $k/n \geq (1-2\rho)^2$, as can be seen by counting the number of generators of the code. Here, $k$ is the number of logical qubits. 

Crucial for the analysis of quantum Tanner codes is the \textit{robustness} of the small tensor codes. Very recently, \cite{DHL22, PK22b} showed that if $C_A$ and $C_B$ are chosen randomly with some fixed rate, then their tensor product is $\kappa$-robust for some constant $\kappa$ independent of the code length $\Delta$. This means that for any \textit{dual tensor codeword} $x \in (C_A^\perp \otimes C_B^\perp)^\perp$, there exist $c \in C_A \otimes \F_2^B, r\in \F_2^A \otimes C_B$ such that $x = c+r$ and
\begin{align}\label{eqn:robust}
|x| \geq \kappa \Delta (\|c\| + \|r\|),
\end{align}
where $|\cdot|$ denotes the Hamming weight and $\|\cdot\|$ counts the number of nonzero columns (resp.~rows) of $c$ (resp.~$r$). Note that the distance of the random codes $C_A$, $C_B$ and their dual will be at least $\delta \Delta$ for some $\delta >0$ with overwhelming probability.
We show in Section~\ref{sec:distance} that these two facts imply a linear lower bound for the distance of the quantum Tanner code.
\begin{theo}\label{thm:dist-main}
For a constant $\Delta$ large enough, the quantum Tanner codes described in the construction above with a Ramanujan left-right Cayley complex have a linear minimum distance:
\[ d_{\min} \geq \frac{ \delta^2\kappa^2}{256 \Delta} n.\]
\end{theo}
The dependency in $\Delta$ of this bound is  tight since there exist logical errors of weight $\leq n/\Delta$ \cite{LZ22}.
The bound is sharper than that of \cite{LZ22}, and its proof is somewhat
simpler.

We note that quantum Tanner codes have found recent applications outside of coding theory: they can fool optimisation algorithms exploiting the sum-of-squares hierarchy \cite{HL22} and are instrumental in the recent proof of the NLTS theorem in quantum complexity theory \cite{FH14,ABN22}.

{\bf Decoding quantum Tanner codes}.---
We focus here on decoding $\sigma_x$-type errors, which are detected by $\sigma_z$-generators. This is without loss of generality for a CSS code. 
The general strategy outlined in \cite{LZ22b} consists in defining a \textit{mismatch} vector associated with the error $\error \in \F_2^Q$, that summarises how the local decoders associated to the local code around each vertex may disagree about the error, and then try to locally modify this mismatch in order to reduce its weight. 
It is natural to see the error $\error$ as a collection of local views on the $Q$-neighbourhoods of
vertices of $V_1$: abusing notation slightly, we can write $\error = \{e_v\}_{v
\in V_1}$, with local views $e_v$ restricted to $Q(v)$. 
For each vertex $v \in V_1$, one can compute (in parallel if needed) a local
error $\eps_v$ with support on $Q(v)$ of minimal Hamming weight yielding the
same local syndrome as $e_v$. This gives a decomposition of the local views of
the error $e_v = \eps_v + c_v + r_v$, with $c_v \in C_A\otimes \F_2^B, r_v =
\F_2^A \otimes C_B$. Note that the $\kappa$-robustness property will apply to any such $c_v+r_v$. 
The issue is that the local views $\{\eps_v\}_{v\in V_1}$ of the decoder are in general not consistent and do not define a global error candidate. 
We measure this inconsistency by defining the mismatch vector:
\begin{equation}
\label{eq:mismatch}
Z := \sum_{v \in V_1} \eps_v\in \F_2^Q.
\end{equation}
If it is equal to zero, it means that each square/qubit is affected the same
value for the two views it belong to, and the decoder is able to define a global
error. Otherwise, the support of $Z$ corresponds to the set of squares for which
the local views of the decoder disagree. On the other hand, the local views $e_v$ of the error \emph{are} consistent, and satisfy $\sum_{v \in V_1} e_v = 0$ since each
square appears twice in this sum. We can rewrite the mismatch \eqref{eq:mismatch} as 
\[
Z = \sum_{v \in V_1} r_v +c_v = C_0 + R_0 + C_1+R_1,
\]
with $C_i := \sum_{v \in V_{\bari i}} c_v$ and $R_j := \sum_{v \in V_{j\barj}} r_v$, where we denote $\bari := 1-i, \barj := 1-j$.
This convenient representation of $Z$ highlights the symmetry of the local
representations of $C_j+R_i$ on vertices of $V_{ij}$, and one may look for local
modifications of $Z$ by adding to it a dual tensor codeword on the
$Q$-neighbourhood of any of the four types of vertices. 

The main subroutine of the decoding algorithm consists in finding a valid
decomposition $Z =\hat{C}_0 + \hat{R}_0 +\hat{C}_1+\hat{R}_1$. The simplest
approach to this is a sequential decoder that proceeds in a greedy fashion:
simply look for a vertex $v$, and local codewords $c_v, r_v$ such that flipping
$c_v+ r_v$ (replacing $Z$ by $Z+c_v+r_v$) decreases the Hamming weight of $Z$. Our main technical result exploits the expansion properties of the square complex together with the robustness of the local codes to establish the existence of some vertex $v\in V$ and local dual tensor codeword $c_v + r_v$ that decreases the weight of $Z$ almost optimally. 
Instrumental in this analysis is the notion of the {\em norm} of a representation of a
mismatch vector $Z$ which can be expressed in many different
ways as $Z=C_0+R_0+C_1+R_1$. The ``column'' vector $C_i$ is expressed as a sum of
local vectors with disjoint supports, each of which is a codeword of $C_A$
supported by a column common to two $Q$-neighbourhoods of two different types of
vertices ($V_{ii}$ and $V_{\bari i}$). The row vectors $R_j$ admit similar
decompositions into $C_B$ components. We will write $\|C_i\|$ ($\|R_j\|$) to denote the
number of non-zero $C_A$ codewords ($C_B$ codewords) in the decomposition of $C_i$ ($R_j$).
We will say that a decomposition $Z=C_0+R_0+C_1+R_1$ is {\em minimal} if it
minimises its norm, namely the
value $\|C_0\|+\|R_0\| +\|C_1\|+\|R_1\|$. Finally, we shall say that a vertex
$v\in V_{ij}$ is
{\em active} for a decomposition $(C_0,R_0,C_1,R_1)$ if $C_j+R_i$ is non-zero
on $Q(v)$.
\begin{theo}
\label{thm:main0}
Fix $\eps \in (0,1)$. 
If, for every $i,j\in\{0,1\}$, the sets of active vertices $S_{ij}\subset V_{ij}$ for a minimum decomposition of a
mismatch vector $Z$ satisfy
\[
|S_{ij}|\leq \frac{1}{2^{12}}\delta^2\eps^3\kappa |V_{00}|
\]
then there exists some vertex $v$ and some codeword $x_v$ of the dual tensor code such that 
\[ |Z| - |Z + x_v| \geq (1-\eps) |x_v|.\]
\end{theo}
Provided the initial error is not too large, we show that the condition on the number of
active vertices in Theorem~\ref{thm:main0} keeps being satisfied, so that
one can keep on
iterating the procedure, finding a local codeword to flip at each step, and that
it will eventually give a decomposition
$Z=\hat{C}_0+\hat{R}_0+\hat{C}_1+\hat{R}_1$ of the mismatch. This decomposition
gives a correct guess
$\hat{\error} = \sum_{v\in V_{00}}\eps_v + \hat{C}_0+\hat{R}_0$ for the error,
that differs from $\error$ by a sum of generators. This naturally yields a
sequential decoder, with parameter $\eps$, described formally later.

It is easy to see that the pre-processing (defining the mismatch) and the post-processing (computing $\hat{\error}$ from the decomposition of $Z$) can be performed in parallel. To get a fully parallel decoder, one simply needs a parallel procedure for decomposing $Z$. 
This is obtained by looking for vertices that would be candidates for the sequential decoder of parameter 
$1/2$, \textit{i.e.}\ finding some $x_v$ such that $ |Z| - |Z + x_v| \geq
|x_v|/2$.
 In order to apply several sequential iterations in parallel, the decoder needs the different $c_v + r_v$'s that
it will identify to have disjoint supports. This is achieved if one restricts
the set of candidate vertices $v$ to a single set $V_{ij}$: so the decoder
applies four parallel substeps, for each of the four sets of vertices
$V_{00},V_{01}, V_{10}, V_{11}$,
each of which consists of applying simultaneously all
sequential iterations for all the candidate vertices it has identified in the
current $V_{ij}$. We will show that after those four substeps, the value of the
current mismatch $|\hat{Z}|$ has been decreased significantly, provided the
original error vector has sufficiently small weight.
The pseudo-code of all the algorithms is presented in
Section~\ref{sec:parallel}.

\begin{theo}
\label{theo:parallel}
Fix $\eps\in(0,1/2)$, and 
$\mu \in (\eps, 1/6)$. 
If the Hamming weight $|\error|$ is less than
\[
\frac{1}{2^{12}}\min\Big(\frac{\eps^3}{16}, \kappa \Big)\delta^2\kappa^2\frac{n}{\Delta},
\]
then the parallel decoder returns a valid correction in logarithmic time. 
\end{theo}
 
The value of $\eps$ can be tuned to what one wants
to achieve. A larger $\eps$ will increase the weight of correctable errors but
increases the number of required parallel steps, while a smaller $\eps$ decreases
the number of parallel steps but only decodes smaller weight errors.

The main technical result needed to prove that the parallel decoder converges
after a logarithmic number of iterations is that there is a linear number of
vertices $v$ that can be updated at each step to decrease the weight of $Z$. The
idea to prove their existence consists in running the sequential algorithm
virtually and then establishing that many of the updates corresponding to the
sequential decoder can in fact be performed at the same time. It is detailed in
the proof of Theorem \ref{corol:parallel} in Section \ref{sub:proof}.

\textbf{Remark:} we have not tried to optimise the power of $2$ numerical
constants in Theorems \ref{thm:dist-main},\ref{thm:main0},\ref{theo:parallel} for the sake
of readability.

{\bf Decoding expander lifted product codes}.---We remark that the same decoding algorithm can be applied to the expander lifted product codes of \cite{PK21}. Indeed, their decoding can be reduced (with a parallel procedure) to that of quantum Tanner codes as explained in \cite{LZ22b}. Furthermore, the decoder also works for hypergraph products of two classical Tanner codes, albeit with smaller correction capabilities since these codes have a distance scaling like $\sqrt{n}$.

{\bf Discussion and open questions}.---
Rapid improvements of decoding algorithms over the last decade have led to a surge of interest for LDPC as a path towards hardware-efficient fault-tolerant quantum computing. The unexpected discovery of good quantum LDPC codes with parameters close to optimal will likely impact this research direction in major ways. The parallel decoding algorithm presented here is a first step towards this goal. Of course, a great number of open questions remains. While our proof requires codes of large size, the decoder is well defined for any quantum Tanner code or lifted product code, and it will be interesting to investigate its performance against random noise. 
A practical decoder should be robust to errors in the syndrome extraction process. We believe this is the case here, and that an analysis along the lines of \cite{FGL18b} could be extended to the present decoder. It also does not suffice to correct errors in order to perform a computation, and one must be able to apply logical gates in a fault-tolerant manner \cite{CKB22}. At the moment, a complete understanding of the logical operators for good LDPC codes is still lacking. Finally, and crucially, we will need examples of good codes of reasonably small size if such codes are to be implemented in real devices.

\section{Technical preliminaries}
\label{sec:prelim}
\subsection{Graph expansion}

In this section, we recall some useful facts about graph expansion. 

Let $\G=(V,E)$ be a graph. Graphs will be undirected but may have multiple
edges. 
For $S,T \subset V$, let $E(S,T)$ denote the multiset of edges with one endpoint in
$S$ and one endpoint in $T$.
Let $\G$ be a connected $\Delta$-regular graph on $n$ vertices, and let
$\Delta=\lambda_1\geq\lambda_2\geq \ldots \geq \lambda_n$ be the eigenvalues of the
adjacency matrix of $\G$. For $n\geq 3$, we define $\lambda(\G):= \max\{|\lambda_i|,
\lambda_i\neq \pm \Delta\}$. 
The graph $\G$ is said to be \emph{Ramanujan} if $\lambda(\G)\leq 2\sqrt{\Delta-1}$.

We recall the following version of the expander mixing lemma (see \textit{e.g.}\
\cite{HLW06}) for bipartite graphs.
\begin{lemma}[Expander mixing lemma] \label{lem:mixing}
Let $\cG$ be a connected $\Delta$-regular bipartite graph on the vertex set $V_0\cup V_1$.
For any pair of sets $S\subset V_0, T \subset V_1$, it holds that
\[ |E(S,T) | \leq \frac{\Delta}{|V_0|} |S| |T| + \lambda(\G) \sqrt{|S| |T|}.\]
\end{lemma}

\subsection{Tanner codes}
\label{subsec:TC}
A binary linear code of length $n$ is an $\f_2$-linear subspace of $\f_2^n$. For
sets $E$ of cardinality $|E|=n$, it will be convenient for us to identify $\f_2^n$
with $\f_2^E$, which we can think of as the space of functions from $E$ to
$\f_2$. Identication with $\f_2^n$ amounts to defining a one-to-one map between $E$ and
$[n]=\{1,2,\ldots ,n\}$, \textit{i.e.}\ a numbering of the elements of $E$.

Let $\G=(V,E)$ be a regular graph of degree $\Delta$, and for any vertex $v$
denote by $E(v)$ the set of edges incident to $v$. 
Assume an identification of $\f_2^{E(v)}$ with $\f_2^\Delta$ for every $v\in V$.
Let $x\in\f_2^{E}$ be a vector indexed by (or a function defined on) the set
$E$. Let us define the {\em local view} of $x$ at vertex $v$ as the subvector
$x_v:=(x_e)_{e\in E(v)}$, \textit{i.e.}\ $x$ restricted to the edge-neighbourhood $E(v)$ of
$v$.

Let $C_0$ be a linear code of length $\Delta$, dimension $k_0=\rho_0\Delta$, and
minimum distance $d_0=\delta_0\Delta$.
We define the Tanner code~\cite{T81} associated to $\G$ and $C_0$ as
\[
\text{Tan}(\G,C_0) :=\{x\in\f_2^E : x_v\in C_0\;\text{for all}\;v\in V\}.
\]
In words, the Tanner code is the set of vectors over $E$ all of whose local
views lie in $C_0$.
By counting the number of linear equations satisfied by the Tanner code, we
obtain
\begin{equation}\label{eq:dimtanner}
\dim \text{Tan}(\G,C_0) \geq (2\rho_0-1)n.
\end{equation}
We also have the bound~\cite{SS96,Gur} on the minimum distance $d$ of the Tanner
code:
\[
d\geq\delta_0 (\delta_0- \lambda(\G)/\Delta)n.
\]
Therefore, if $(\G_i)$ is a family of $\Delta$-regular expander graphs with
$\lambda(\G_i)\leq\lambda <d_0$, and if $\rho_0>1/2$, then the associated family of Tanner
codes has rate and minimum distance which are both $\Omega(n)$, meaning we have
an asymptotically good family of codes, as was first shown in~\cite{SS96}.

\subsection{Quantum CSS codes}

A quantum CSS code is specific instance of a stabilizer code~\cite{got97} that can be defined by two classical codes $\C_0$ and $\C_1$ in the ambient space $\f_2^n$, with the property that $\C_0^\perp \subset \C_1$~\cite{CS96,ste96}. It is a \emph{low-density parity-check} (LDPC) code whenever both $\C_0$ and $\C_1$ are the kernels of sparse parity-check matrices. 
The resulting quantum code $\eQ = (\C_0, \C_1)$ is a subspace of
$(\mathbb{C}_2)^{\otimes n}$, the space of $n$ qubits:
\[ \eQ := \mathrm{Span}\left\{ \sum_{z \in \C_1^\perp} |x+z\rangle \: : \: x\in \C_0 \right\},\]
where $\{ |x\rangle \: : \: x\in \f_2^n\}$ is the canonical basis of $(\mathbb{C}_2)^{\otimes n}$.
The dimension $k$ of the code counts the number of logical qubits and is given by 
\[ k = \text{dim} \, (\C_0/\C_1^\perp) = \text{dim} \, \C_0 + \text{dim} \, \C_1 - n.\]
Its minimum distance is $d = \min (d_X, d_Z)$ with
\[ d_X = \min_{w \in \C_0 \setminus \C_1^{\perp}} |w|, \quad d_Z = \min_{w \in \C_1 \setminus \C_0^\perp} |w|.\]
We denote the resulting code parameters by $\llbracket n,k,d\rrbracket$ and say
that a code family $(\eQ_n)_n$ is \emph{asymptotically good} if its
parameters are of the form
\[ \llbracket n, k = \Theta(n),d = \Theta(n)\rrbracket.\]

An $n$-qubit Pauli error $E_1 \otimes \ldots \otimes E_n$ with $E_i \in \{ \1, \sigma_X, \sigma_Y, \sigma_Z\}$\footnote{The 1-qubit Pauli matrices are defined by $\1 = \left( \begin{smallmatrix} 1 &0\\0&1 \end{smallmatrix}\right), \sigma_X= \left( \begin{smallmatrix} 0 &1\\1&0 \end{smallmatrix}\right), \sigma_Z= \left( \begin{smallmatrix} 1 &0\\0&-1 \end{smallmatrix}\right)$ and $\sigma_Y = i \sigma_X \sigma_Z$.} is conveniently described by two $n$-bit strings $(e_0, e_1)\in \F_2^n \times \F_2^n$ \textit{via} the mapping 
\[ \1 \mapsto (0,0), \quad \sigma_X \mapsto (1,0), \quad \sigma_Y \mapsto (1,1), \quad \sigma_Z \mapsto (0,1),\]
which forgets global phases. 
The parity-check matrices of $\C_0$ and $\C_1$ give rise to syndrome maps
$\sigma_0, \sigma_1 : \F_2^n \to \F_2^m$ that associate a pair of syndromes
$(\sigma_0(\error_0), \sigma_1(\error_1)) \in \F_2^m \times \F_2^m$ to any
$n$-qubit Pauli error $(\error_0, \error_1)\in \F_2^n \times \F_2^n$.
The decoding problem for a stabilizer code is as follows: given a syndrome
$(\sigma_0(\error_0), \sigma_1(\error_1))$, recover the error up to an element
of the stabilizer group, that is return $(\hat{\error}_0, \hat{\error}_1)$ such
that $\error_0 + \hat{\error}_0 \in \C_1^\perp$ and $\error_1 + \hat{\error}_1 \in \C_0^\perp$. 

While an optimal decoding of \emph{random} errors would typically exploit
possible correlations between $\error_0$ and $\error_1$, it is always possible
to correct both errors independently. Here, we will be concerned with the
adversarial setting where $\error_0$ and $\error_1$ are of sufficiently low
weight, but otherwise arbitrary. In that case, both errors should be decoded
independently, and we will focus on the case where $(\error_0=0, \error_1=e)$ in
the sequel.

\subsection{Left-right Cayley complexes (quadripartite version)}
\label{subsec:LRCayley}

The square complex we shall rely on for the construction first appeared in \cite{PK21} as a balanced product of double covers of non-bipartite Cayley graphs. For the sake of simplicity, we will rather use the language of left-right Cayley complexes in their quadripartite version.
A {\em left-right Cayley complex} $X$ is introduced in~\cite{DEL22} from a group
$G$ and two sets of generators $A=A^{-1}$ and $B=B^{-1}$. As in~\cite{DEL22} we
will restrict ourselves, for the sake of simplicity, to the case
$|A|=|B|=\Delta$. The complex is made up of vertices, $A$-edges, $B$-edges, and
squares. The vertex set consists of four copies of the group $G$ in the
quadripartite version, $V = V_{00} \cup V_{10} \cup V_{01}\cup V_{11}$ with
$V_{ij} = G \times \{ij\}$. 
The advantage of this quadripartite version, also considered in \cite{PK21} and \cite{gol21}, is that it does not require any additional assumption on the choice of group and generators, for instance that $ag \ne gb$ for all $g\in G, a\in A, b \in B$, as in \cite{DEL22}.
We will also use the notation $V_0:=V_{00}\cup
V_{11}$ and $V_1:=V_{01}\cup V_{10}$.
The $A$-edges are pairs of vertices of the form $\{(g,i0) ,(ag,i1)\}$ and $B$-edges are of
the form $\{(g,0j),(gb,1j)\}$ for $g\in
G,a\in A,b\in B$, $i,j=0,1$. We denote by $E_A$ and $E_B$ these two edge sets. The associated graphs are denoted by $\G_A = (V, E_A)$ and $\G_B = (V,E_B)$.
 A {\em square} is a set of four vertices of the form $\{(g,00),(ag,01),(gb,10),(agb,11)\}$.
The set of squares (or quadrangles) of the complex is denoted by $Q$.
Every vertex is incident to exactly $\Delta^2$ squares. For a vertex $v$, the set of incident squares is called the $Q$-neighbourhood,
and denoted by $Q(v)$.

The sets of generators $A$ and $B$ will be chosen so that the Cayley graphs
$\Cay(G,A)$ and $\Cay(G,B)$ are non-bipartite Ramanujan graphs.
It should be understood that when writing $\Cay(G,A)$ we implicitely mean the
Cayley graph defined by left multiplication by
elements of $A$, while $\Cay(G,B)$ stands for the Cayley graph defined by
right multiplication by elements of $B$. The sets $A$ and $B$ could in principle
be chosen to be identical, but we keep a distinct notation for both sets, in
particular in order to allow the above abuse of notation to be non-confusing.

We see that the subset of edges of $E_A$ that connect vertices of $V_{00}$ to
vertices of $V_{01}$ make up a double cover of the Cayley graph $\Cay(G,A)$,
the edges $E_A$ that connect $V_{10}$ to $V_{11}$ make up a second copy of the
same double cover. Therefore, the graph $\G_A$ is a disjoint union of two copies
of the double cover of $\Cay(G,A)$. Similarly, $\G_B$ is a disjoint union of two
copies of the double cover of $\Cay(G,B)$.

Let us introduce one additional graph that exists on the complex $X$,
and that we denote by $\G^\square$. This graph
puts an edge between all pairs of vertices of the form $\{(g,i),(agb,i)\}$, $g\in G,a\in A,b\in
B, i=0,1$. The graph $\G^\square$ is therefore made up of two connected
components, on $V_0$ and $V_1$, that we denote by $\G_0^\square$ and
$\G_1^\square$. We note that $\G^\square$ is regular of degree $\Delta^2$, and
may have multiple edges. 

If $\Cay(G,A)$ and $\Cay(G,B)$ are Ramanujan, then $\G^\square$ inherits some
of their expansion properties. Specifically:

\begin{lemma}\label{lem:lambda}
Assume that $\Cay(G,A)$ and $\Cay(G,B)$ are Ramanujan graphs, then 
\[ 
\lambda(\G_0^\square) \leq 4\Delta, \quad \lambda(\G_1^\square)\leq 4\Delta.\]
\end{lemma}

The proof follows from the fact that the adjacency matrix of $\G^\square$ is the
product of the adjacency matrices of $\G_A$ and $\G_B$, and that these two
adjacency matrices commute, by definition of the square complex. See~\cite{LZ22}
for a little more detail.

\subsection{Labelling $Q$-neighbourhoods}
We will define Tanner codes on $\G_0^\square$ and $\G_1^\square$, which implies a labelling of
the coordinates in every $Q$-neighbourhood $Q(v)$. There is a natural labeling
of $Q(v)$ by the set $A\times B$, namely a one-to-one map $\phi_v~: A\times B
\to Q(v)$,  which we now state explicitely.

We set
\begin{align*}
\text{for}\; v=(g,00)\in V_{00},\quad& \phi_v(a,b) =
\{(g,00),(ag,01),(gb,10),(agb,11)\},\\
\text{for}\; v=(g,01)\in V_{01},\quad& \phi_v(a,b) =
\{(g,01),(a^{-1}g,00),(gb,11),(a^{-1}gb,10)\},\\
\text{for}\; v=(g,10)\in V_{10},\quad& \phi_v(a,b) =
\{(g,10),(ag,11),(gb^{-1},00),(agb^{-1},01)\},\\
\text{for}\; v=(g,11)\in V_{11},\quad& \phi_v(a,b) =
\{(g,11),(a^{-1}g,10),(gb^{-1},01),(a^{-1}gb^{-1},00)\}.
\end{align*}
The map $\phi_v$ thus defined is obviously one-to-one, and one easily checks
that:

\medskip

{\em
Any two vertices $v=(g,0i)$ and $w=(gb,1i)$, $i=0,1$, that are connected through
a $B$-edge (labelled $b$), have a common ``column'', \textit{i.e.}\ their
$Q$-neighbourhoods share exactly $\Delta$ squares that are labelled $(a,b), a\in
A$, in both $Q(v)$ and $Q(w)$.
}
\medskip

Similarly,

\medskip

{\em
Any two vertices $v=(g,i0)$ and $w=(ag,i1)$, $i=0,1$, that are connected through
an $A$-edge (labelled $a$), have a common row, \textit{i.e.}\ their
$Q$-neighbourhoods share exactly $\Delta$ squares that are labelled $(a,b), b\in
B$, in both $Q(v)$ and $Q(w)$.
}

\medskip

The situation is illustrated on Figure~\ref{fig:code}. Summarising, any two
vertices connected by $B$-edge (an $A$-edge) have a common column (row) in their $Q$-neighbourhoods,
that is labelled by the same $b\in B$ ($a\in A$).

\subsection{Local codes}

The constraints of a classical Tanner code consist of local constraints from small codes enforced on the edge-neighbourhood of each vertex. 
For quantum Tanner codes, now that all local $Q$-neighbourhoods are isomorphic
to $A \times B$, we may put local constraints that are codewords of the tensor
codes $C_A \otimes C_B$ and $C_A^\perp \otimes C_B^\perp$.

Recall that the generators of the quantum Tanner code correspond to a basis of $C_A \otimes C_B$ on each $Q$-neighbourhood 
of $V_{00} \cup V_{11}$ (for the $\sigma_X$-type generators) and to a basis of $C_A^\perp \otimes C_B^\perp$ 
on each $Q$- neighbourhood of $V_{01} \cup V_{10}$ for the $\sigma_Z$-type generators). 
The classical code $\C_0\subset \F_2^Q$ correcting $\sigma_Z$-type errors is the Tanner code on the graph $\G_0^\square$ with local constraints corresponding to the dual tensor code $(C_A \otimes C_B)^\perp = C_A^\perp \otimes \F_2^B + \F_2^A \otimes C_B^\perp$.
With the notation of Section~\ref{subsec:TC}, $\C_0=\text{Tan}(\G_0^\square,
C_A^\perp \otimes \F_2^B + \F_2^A \otimes C_B^\perp)$.
Similarly, the classical code $\C_1\subset\F_2^Q$ correcting $\sigma_X$-type errors is the Tanner
code on the graph $\G_1^\square$ with local constraints corresponding to the
dual tensor code $(C_A^\perp \otimes C_B^{\perp})^\perp = C_A \otimes \F_2^B +
\F_2^A \otimes C_B$, \textit{i.e.}\ $\C_1=\text{Tan}(\G_1^\square, C_A \otimes \F_2^B +
\F_2^A \otimes C_B)$.

\medskip

\noindent
\textbf{Summary of parameters.}
A large enough $\Delta$ is chosen, together with an infinite family of groups
$G$ with generating sets $A,B$, $|A|=|B|=\Delta$, such that the left Cayley
graph $\Cay(G,A)$ and the right Cayley graph $\Cay(G,B)$ are Ramanujan. The
quadripartite left-right square complex $X$ is defined by $G,A,B$.

We take inner codes $C_A$ and $C_B$ such that $\dim C_B=\Delta-\dim C_A$.
Defining the inner rate $\rho$ such that $\dim C_A=\rho\Delta$, and
counting the number of constraints, we obtain that the dimension of the quantum
Tanner code $\eQ=(\C_0,\C_1)$ of length $|Q|$ is at least
$(1-2\rho)^2|Q|$.

\subsection{Robustness of (dual) tensor codes}

Besides the expansion properties of the left-right Cayley complex, the other
required property to obtain good quantum LDPC codes is the robustness of the
local dual tensor codes \cite{PK22b,LZ22}. In words, it states that any low-weight codeword of the dual tensor code $C_A \otimes \F_2^B + \F_2^A \otimes C_B$ can be obtained as a sum $c_v + r_v$ where the Hamming weights of $c_v$ and $r_v$ are not much larger than that of $c_v+r_v$.
Very recently, better bounds on the robustness of dual tensor codes obtained
from two random codes were obtained in \cite{GPT22,DHL22,PK22b}. In particular, \cite{DHL22,PK22b} prove essentially tight bounds. 
We recall the main result from \cite{PK22b}, where robustness is called
product expansion and is defined as follows:  for two linear codes
$C_A,C_B\subset\F_2^\Delta$, the dual tensor code $C_A \otimes
\F_2^\Delta + \F_2^\Delta \otimes C_B$ is said to be {\em $\kappa$-product
expanding}, if any codeword $x \in C_A \otimes \F_2^\Delta + \F_2^\Delta
\otimes C_B$ can be written as $x = c+r$ with $c \in  C_A \otimes \F_2^\Delta, r
\in  \F_2^\Delta \otimes C_B$ and 
\[ |c+r| \geq \kappa \Delta (\|c\| + \|r\|),\]
where $\|c\|$ ($\|r\|$) denotes the number of columns (rows) involved in the
support of $c$ ($r$).
\begin{theo}[\cite{PK22b}]\label{thm:robust}
For every $\rho_A, \rho_B \in (0,1)$, the dual tensor code $C_A \otimes
\F_2^\Delta + \F_2^\Delta \otimes C_B$ obtained from a uniformly random pair of
linear
codes $(C_A,C_B)$, of lengths $\Delta$ and of codimensions $\lceil \rho_A \Delta\rceil$ and $\lceil
\rho_B \Delta \rceil$ respectively,
is $\kappa$-product-expanding with high probability as $\Delta \to \infty$ with
\[ \kappa = \frac{1}{2} \min \left(
\frac{1}{4}H_2^{-1}\left(\frac{\rho_A}{8}\right)H_2^{-1}\left(\frac{\rho_B}{8}\right),
H_2^{-1}\left(\frac{\rho_A \rho_B}{8}\right)\right).\]
\end{theo}
Here $H_2^{-1}$ is the inverse of the binary entropy function given by 
\[ H_2(x) := -x \log_2 x - (1-x) \log_2(1-x).\]

\noindent
{\bf Remark:} for our application we have $\rho_B=1-\rho_A$. Therefore,
Theorem~\ref{thm:robust} implies that both dual tensor codes $C_A \otimes
\F_2^\Delta + \F_2^\Delta \otimes C_B$ and $C_A^\perp \otimes
\F_2^\Delta + \F_2^\Delta \otimes C_B^\perp$ are with high probability
$\kappa$-product-expanding for the same value $\kappa$.

\section{Minimum distance}
\label{sec:distance}
Let $\x$ be a codeword of $\C_1$ not in $\C_0^\perp$. We will derive a lower
bound on the weight of $\x$, that will also be valid for the weight of
a codeword of $\C_0$ not in $\C_1^\perp$, and therefore bound from below the minimum distance of the code.

The codeword $\x$ is a Tanner codeword on the graph $\G_1^\square$, with vertex
set $V_1=V_{01}\cup V_{10}$ and the dual tensor code
$C_A\otimes\F_2^B+\F_2^A\otimes C_B$ as inner code.
The set $Q$ of coordinates is partitioned into $Q$-neighbourhoods $Q(v)$ of
vertices $v$  of $V_{01}$, on which $\x$ reduces to a dual-tensor codeword
$x_v=c_v+r_v$, with $c_v$ and $r_v$ being $C_A\otimes\f_2^B$ codewords and
$\f_2^A\otimes C_B$ codewords respectively.
We regroup the ``column'' vectors and ``row'' vectors and write $C_1=\sum_{v\in
V_{01}}c_v$ and $R_0=\sum_{v\in V_{01}}r_v$. Similarly, using the partition of $Q$ into $Q(v)$s for $v\in
V_{10}$, we define $C_0=\sum_{v\in
V_{10}}c_v$ and $R_1=\sum_{v\in V_{10}}r_v$.

Let us denote by $\|C_i\|=\sum_{v\in V_{\bari i}}\|c_v\|$ the total number of non-zero column
vectors in $C_A$ intervening in the local views of $\x$, for the partition over
$Q(v), v\in V_{\bari i}$.  We write $\bari=1-i$ and $\barj=1-j$ to lighten
notation. Similarly, we denote by $\|R_i\|$ the quantity
$\|R_i\|=\sum_{v\in V_{i\bari}}\|r_v\|$. Note that there are several possible
decompositions of a local view $x_v$ of $\x$ as $x_v=c_v+r_v$. If for every
$v\in V_1$ we choose one that minimises $\|c_v\|+\|r_v\|$, that is the total
number of $C_A$ and $C_B$ codewords in the decomposition of $x_v$, we will
obtain a {\em minimal} representation $(C_0,R_0,C_1,R_1)$ of $\x$ that minimises
the quantity $\|C_0\|+\|R_0\|+\|C_1\|+\|R_1\|$. We call the latter quantity the {\em norm} of
$\x$, and denote it by
\[
\|\x\|:=\|C_0\|+\|R_0\|+\|C_1\|+\|R_1\|.
\]
Since $\x=\sum_{v\in V_{01}}x_v=C_1+R_0=\sum_{v\in V_{10}}x_v=C_0+R_1$, we have
$C_0+R_0+C_1+R_1=0$, and therefore $C_0+R_0=C_1+R_1$. We make the remark that
the local views of the vector $\x^0:=C_0+R_0=C_1+R_1$ at vertices of $V_0$ are also dual
tensor codewords in $C_A\otimes\F_2^B+\F_2^A\otimes C_B$ (which should not be
confused with the local views of $\x$ on vertices of $V_0$). In other words the
vector $\x^0$ is a codeword of the Tanner code on the graph $\G_0^\square$ and
the same inner code $C_A\otimes\F_2^B+\F_2^A\otimes C_B$. Now, if $(C_0,R_0,C_1,R_1)$ is a minimal representation for
$\x$, it may not necessarily be a minimal representation for $\x^0$. 
However, if we consider a local view $x_v^0$ of $\x^0$ at $v\in V_{00}$, and 
if its decomposition $x_v^0=c_v+r_v$ is not minimal, where $r_v$ and $c_v$ are
the local views at $v$ of $R_0$ and $C_0$, then we may replace the local
decomposition by a minimal one, which will equal $x_v^0=(c_v+t_v)+(r_v+t_v)$,
where $t_v$ is some tensor codeword in $(C_A\otimes \f_2^B)\cap(\f_2^A\otimes C_B)$.
This has the effect of changing $C_1+R_0$ and $C_0+R_1$ to $C_1+R_0+t_v$ and
$C_0+R_1+t_v$, and of reducing the sum $\|C_0\|+\|R_0\|+\|C_1\|+\|R_1\|$.
We observe that in this case $\x^0$ is unchanged and $\x$ is changed to another
codeword of $\C_1$, which stays in the same class $\x+\C_0^\perp$, since the
tensor codeword $t_v$ is a generator. The conclusion is that we can keep
proceeding in this way until no local modification is possible, at which point
we will have replaced $\x$ by a codeword of the same
class modulo $\C_0^\perp$, and $(C_0,R_0,C_1,R_1)$ will be a minimal
representation of both $\x$ and of $\x^0$. We have shown in particular:
\begin{lemma}
\label{lem:minimal}
If $\x\in\C_1\setminus C_0^\perp$, and $\|\x\|$ is the minimum norm in the coset
$\x+\C_0^\perp$, then a minimal representation for $\x$ is also a minimum
representation for $\x^0$.
\end{lemma}

We will prove the following lower bound on the norm of a codeword of
$\C_1\setminus\C_0^\perp$.
\begin{lemma}
\label{lem:norm}
For any $\x\in \C_1\setminus\C_0^\perp$, we have 
\[
\|\x\| \geq
\frac{\delta^2\kappa  n}{512\Delta^2}.
\]
\end{lemma}

From Lemma~\ref{lem:norm} we easily deduce a lower bound on the quantum minimum
distance.

\begin{theo}\label{thm:dist}
The minimum distance of the quantum Tanner code satisfies
\[
d_{\min} \geq
\frac{ \delta^2 \kappa^2 n}{256\Delta}.
\]
\end{theo}

\begin{proof}
Let $\x$ be a codeword of $\x\in \C_1\setminus\C_0^\perp$, and let
$(C_0,R_0,C_1,R_1)$ be a minimal representation for $\x$. We
have
\[
\x = C_1+R_0=\sum_{v\in V_{01}}(c_v+r_v).
\]
Since the local vectors in this sum have disjoint supports, we have
\[
|\x|=\sum_{v\in V_{01}}|c_v+r_v|\geq \sum_{v\in V_{01}}\kappa\Delta
(\|c_v\|+\|r_v\|)
\]
applying robustness of the local codes and minimality of the representation.
Hence $|\x|\geq\kappa\Delta(\|C_1\|+\|R_0\|)$, and similarly 
$|\x|\geq\kappa\Delta(\|C_0\|+\|R_1\|)$ by summing over $V_{10}$. Therefore,
\[
|\x|\geq \kappa\Delta\frac 12(\|C_1\|+\|R_0\|+\|C_0\|+\|R_1\|)=\kappa\Delta\frac
12\|\x\|
\]
which proves the lower bound for non-trivial codewords of $\C_1$. The same lower
bound holds for non-trivial codewords of $\C_0$ by symmetry.
\end{proof}

It remains to prove Lemma~\ref{lem:norm}. We may suppose that $\x$ achieves the
minimum of $\|\x\|$ in $\x+\C_0^\perp$, so that Lemma~\ref{lem:minimal} holds.
In the following, $(C_0,R_0,C_1,R_1)$ is a minimal representation of $\x$.

Let us denote $S_{ij}$ the set of vertices $v$ of $V_{ij}$ for which $C_j+R_i$ is
non-zero on $Q(v)$. Let us also set $S_0= S_{00}\cup S_{11}$, and
$S_1= S_{01}\cup S_{10}$. 

Let us call a vertex $v$ of $V_{ij}$ {\em exceptional}, if $\|c_v\|+\|r_v\|$
is at least $\alpha\Delta$ with $\alpha =\delta^2/256$, where $c_v$ and
$r_v$ are the restrictions to $Q(v)$ of $C_j$ and $R_j$ respectively. Then we write $S_{ij}^e \subset S_{ij}$ the set of exceptional vertices in $S_{ij}$.

{\em Ordinary} vertices of $S_{ij}$ are defined as non-exceptional.
We remark that any non-zero column $C_A$-vector of $C_j$ (or any row vector of $R_i$) in the
$Q$-neighbourhood of any ordinary vertex of $S_{ij}$ has, by definition of
ordinary, at most $\frac{\delta}{256}\delta\Delta$ non-zero coordinates in common with
$R_i$, which we bound from above by $\frac 12\delta\Delta$ for the sake of
readability.

The following lemma states that whenever $S_{ij}$ is small enough, the number of
exceptional vertices is a small fraction of $|S_{ij}|$, of order
$O(1/\Delta^2)$, hence the terminology.

\begin{lemma}
\label{lem:excep}
Let $i=0,1$. Under the hypothesis $|S_{ij}|\leq \frac{\alpha\kappa}{2}|V_{00}|$, we have 
\[
|S_{ii}^e|\leq \frac{64}{\alpha^2\kappa^2\Delta^{2}}|S_{0}|
\quad\text{and}\quad |S_{i\bari}^e|\leq \frac{64}{\alpha^2\kappa^2\Delta^{2}}|S_{1}|.
\]
\end{lemma}
\begin{proof}
We prove the upper bound for $S_{00}^e$, the other cases being similar. 
Viewing $C_0+R_0$ as a subgraph of $\G_0^\square$, we have that vertices of
$S_{00}^e$ have degree at least $\kappa\Delta\alpha\Delta$ (applying
robustness, which we may do by minimality of $\|C_0\|+\|R_0\|+\|C_1\|+\|R_1\|$), and the Expander mixing Lemma in $\G_0^\square$ gives
\[
|S_{00}^e|\alpha\kappa\Delta^{2}\leq |E(S_{00}^e,S_{11})|\leq
\Delta^2\frac{|S_{00}^e||S_{11}|}{|V_{00}|} + 4\Delta\sqrt{|S_{00}^e||S_{11}|}.
\]
Upper bounding $|S_{11}|/|V_{00}|$ by $\kappa\alpha/2$, we obtain
\[
\frac 12|S_{00}^e|\Delta^{2}\alpha\kappa\leq 4\Delta\sqrt{|S_{00}^e||S_{11}|} 
\]
and finally
\[
|S_{00}^e| \leq \frac{64}{\alpha^2\kappa^2\Delta^{2}}|S_{11}|\leq
\frac{64}{\alpha^2\kappa^2 \Delta^{2}}|S_{0}|. \qedhere
\]
\end{proof}

Let $v$ be a vertex of $V_{ij}$ and consider a column on its $Q$-neighbourhood. 
Suppose this column supports a non-zero $C_A$-codeword that is part of $C_j$.
Recall that this column is shared by the  $Q$-neighbourhood of a neighbouring vertex $w$ in $V_{\bari j}$.
Below we consider the set $T$ of vertices of $V_{ij}$ whose local views see at
least one
non-zero $C_A$-codeword of $C_j$ that is shared by the local view of an {\em
ordinary} vertex $w\in V_{\bari j}$.
\begin{lemma}
\label{lem:T}
If we have $|S_{i\barj}|\leq
\delta|V_{00}|/4$, then 
\[
|T|\leq \frac{64}{\delta^2\Delta}|S_{i\barj}|.
\]
Furthermore,
exactly the same result holds if $T$ is defined as the subset of vertices of
$S_{\bari\barj}$ on whose $Q$-neighbourhood
$R_{\bari}$ displays a non-zero row codeword of $C_B$ that is shared by the local view
of an ordinary vertex $w\in S_{\bari j}$. 
\end{lemma}

\begin{proof}
We deal with the case when $i=j=0$ and $T$ is defined as the set of vertices
of $V_{00}$ whose $Q$-neighbourhoods share a non-zero column vector with an
ordinary vertex of
$V_{10}$. The other cases will hold by symmetry.
If we keep only the two sets of vertices $V_{00}$ and $V_{01}$, then every
square of the square complex becomes incident to two vertices (instead of four), 
and we obtain a multigraph. In this multigraph, every row of a local view
appears in two $Q$-neighbourhoods, one of a vertex $v\in V_{00}$, and the other in a
neighbouring vertex of $V_{01}$. If we collapse this row to a single ``square''
(which has now become an edge) by identifying them, the bipartite graph over
$V_{00}\cup V_{01}$ that we
obtain in this way is exactly the double cover of the Cayley graph $\Cay(G,A)$.
Now the codeword $\x$ induces a subgraph of this graph, which we obtain by
putting an edge in the subgraph whenever the row view it originates from is
non-zero in $\x$.

By construction, every vertex of $T$ in this subgraph has degree at least $\delta\Delta/2$ 
(actually degree at least $(1-\frac{\delta}{256})\delta\Delta$, as discussed
above), and its edges fall into $S_{01}$. Hence, 
\[
|T|\frac{\delta\Delta}{2}\leq |E(S_{01},T)|.
\] 
Applying the expander mixing Lemma in the double cover of $\Cay(G,A)$ we obtain,
\[
|E(S_{01},T)|\leq \Delta\frac{|S_{01}||T|}{|V_{01}|} +
2\sqrt{\Delta}\sqrt{|S_{01}||T|}.
\]
Finally, applying the hypothesis $|S_{01}|\leq\delta|V_{01}|/4$ we get
\[
|T|\frac{\delta\Delta}{4}\leq 2\sqrt{\Delta}\sqrt{|S_{01}||T|}
\]
and the result follows. We remark that when $T$ is defined in $V_{i\bari}$
as opposed to $V_{ii}$, it is the codeword $\x^0=R_0+C_0$ that is considered rather than
$\x$.
\end{proof}

Now, Lemma~\ref{lem:excep} will tell us that if $\|\x\|$ is too small, there can
only be very few exceptional vertices. But what Lemma~\ref{lem:T} then tells us,
is that the non-zero column codewords in the neighbourhoods of ordinary
vertices must cluster in a
limited number of $Q$-neighbourhoods for vertices of a neighbouring type,
yielding too many exceptional vertices and contradicting Lemma~\ref{lem:excep}. We now make this argument formal.

\begin{proof}[Proof of Lemma~\ref{lem:norm}]
Suppose $\x\in \C_1\setminus\C_0^\perp$ satisfies $\|\x\|< \frac{\kappa \delta^2
n}{512\Delta^2}$. 
We may suppose that $\|\x\|$ achieves its minimum value inside the coset
$\x+\C_0^\perp$, and that $(C_0,R_0,C_1,R_1)$ is a minimal representation for
$\x$.
Since we have $|S_{ij}|\leq\|C_j\|+\|R_i\|\leq\|C_0\|+\|R_0\|+\|C_1\|+\|R_1\|=\|\x\|$, 
we have that $|S_{ij}|\leq \frac{\kappa\delta^2}{512}|V_{00}|=\frac
12\kappa\alpha|V_{00}|$. Therefore Lemma~\ref{lem:excep} holds, and so does 
Lemma~\ref{lem:T} since clearly $\frac 12\kappa\alpha \leq \delta/4$.

Without loss of generality let us suppose $|S_1|\geq |S_0|$ (otherwise invert
their roles in the argument below) and $|S_{10}|\geq |S_{01}|$ (otherwise invert
their roles).

Because there are so few exceptional vertices in $S_{10}$ (by
Lemma~\ref{lem:excep}), the number
of ordinary vertices of $S_{10}$ is almost equal to $|S_{10}|$ namely it is at
least $a|S_{10}|$ for any constant $a<1$ and $\Delta$ large enough.

We obtain that either the number of ordinary rows or the number of ordinary
columns in (the $Q$-neighbourhoods of vertices of) $S_{10}$ is at least $a|S_{10}|/2$. 
Suppose without loss of generality\footnote{If it is the number of
ordinary rows that exceeds $a|S_{10}|/2$, then $T$ is defined inside $S_{11}$
instead of $S_{00}$ and the rest of the argument is unchanged.}
that the number of ordinary columns is at
least $a|S_{10}|/2$.
Applying Lemma~\ref{lem:T}, they must cluster among the $Q$-neighbourhoods of a
set $T\subset S_{00}$ of size
at most $\frac{64}{\delta^2\Delta}|S_{01}|\leq
\frac{64}{\delta^2\Delta}|S_{10}|$ (since we have supposed $|S_{01}|\leq
|S_{10}|$).
Therefore, the average number of non-zero columns of $C_0$ on the
$Q$-neighbourhoods of the vertices of this set
$T$ is at least $a |S_{10}|/2 \big(\frac{64}{\delta^2\Delta}|S_{10}\big)^{-1} = a\frac{\delta^2\Delta}{128}$, which is close to twice the
minimum norm of a local view of $R_i+C_j$ for an exceptional vertex.
This shows that a constant proportion of vertices of $T$ must be exceptional vertices
of $S_{00}$. Now since the $Q$-neighbourhood of a vertex can host at most
$\Delta$ columns, we also have $|T|\geq \frac 1\Delta \frac a2|S_{10}|\geq \frac
a{4\Delta} |S_1|$ (since $|S_{10}|\geq |S_{01}|$), and since $|S_1|\geq |S_0|$,
we have $|T|\geq \frac{a}{4\Delta}|S_0|$. Therefore, for large enough $\Delta$,
we have a contradiction with Lemma~\ref{lem:excep} which limits the number of
exceptional vertices of $S_{00}$ to not more than $O(1/\Delta^2)|S_0|$.
\end{proof}

\section{Description of the decoders}
\label{sec:parallel}

In this section, we present in detail the various decoders: the general decoding procedure including the computation of the mismatch and the post-processing is described in Algorithm \ref{algo:decoder}.
We give two procedures for finding a mismatch decomposition of the form $Z =
\hat{C}_0 + \hat{R}_0 + \hat{C}_1 + \hat{R}_1$, a sequential procedure described
in Algorithm \ref{algo:seq-mismatch} and a parallel one described in Algorithm
\ref{algo:mismatch}. 

Both the sequential decoder and the parallel decoder follow the blueprint of
Algorithm \ref{algo:decoder}. They use a {\em preprocessing phase} which
consists of computing a mismatch vector $Z$. Specifically, on the
$Q$-neighbourhood of every vertex of $V_0$ (if one is trying to correct a
$\sigma_z$-error), or of $V_1$ (if one is trying to correct a $\sigma_x$-error),
the decoder computes a local estimation $\eps_v$ of the original error vector
$\error$, and then sums all these $\eps_v$'s to make up the mismatch vector $Z$.
Note that this common preprocessing phase can be parallelised straightforwardly
when needed.

The core of the decoding algorithm is then to uncover a decomposition of the
mismatch vector of the form $Z=\hat{C}_0+\hat{R}_0+\hat{C}_1+\hat{R}_1$, where 
$\hat{C}_i$ is a sum of local (column) codewords of $C_A$ whose coordinates are indexed by a column of a
$Q$-neighbourhood $Q(v)$ for $v\in V_{ii}$, and similarly $\hat{R}_i$ is a sum of local
row codewords of $C_B$ whose coordinates are indexed by a row in some
$Q(v)$, $v\in V_{ii}$. Note that the individual local column codewords of $C_i$
also appear in the $Q$-neighbourhoods of vertices of $V_{\bari i}$, and the
row codewords of $R_i$ appear likewise in $Q$-neighbourhoods of $v\in
V_{i\bari}$. The mismatch decomposition procedures differ significantly for the
sequential decoder and the parallel decoder.

Finally, both decoders use a common postprocessing phase once they have decomposed the
mismatch vector $Z$. If one is decoding a $\sigma_z$-error, the decoders output
their estimation of the error which is computed as
\[
\hat{\error} = \sum_{v\in V_{00}}\eps_v + \hat{C}_0+\hat{R}_0.
\]
We remark that we also have
\[
\hat{\error} = \sum_{v\in V_{11}}\eps_v + \hat{C}_1+\hat{R}_1
\]
since the sum of those two quantities is zero by construction and by definition
of the mismatch.
Similarly, if one is decoding a $\sigma_x$-error, the decoders compute
$\hat{\error}=\sum_{v\in V_{10}}\eps_v + \hat{C}_0+\hat{R}_1=\sum_{v\in
V_{01}}\eps_v + \hat{C}_1+\hat{R}_0$.
Since these summations can be done locally on the components of the partition of $Q$ into
$Q$-neighbourhoods of the relevant $V_{ij}$, the postprocessing is 
achieved naturally by means of a parallel computation.

It remains to describe the mismatch decomposition procedures.
The sequential mismatch decomposition procedure is given 
in Algorithm~\ref{algo:seq-mismatch}.
The sequential procedure is rather natural: it consists of initiating a variable
$\hat{Z}$ at $Z$, and iteratively
looking for a vertex
$v\in V$ and a non-zero dual tensor codeword $x_v$ supported by $Q(v)$, such that adding $x_v$
to $\hat{Z}$ decreases its weight by a sufficient amount. The decoder then
decomposes $x_v$ as $x_v=c_v+r_v$ with $\|c_v\|+\|r_v\|$ minimum, and increments
$\hat{C}_j$ by $c_v$ and $\hat{R}_i$ by $r_v$, where $i,j$ are such that $v\in V_{ij}$.
Iterations continue until $\hat{Z}=0$.

The parallel decomposition procedure uses the fact that the local views of
vertices of a given $V_{ij}$ have disjoint supports, and one may therefore
update these local views in parallel without ever risking conflicting
instructions. It follows that one round of parallel decoding will consist of four
consecutive parallel substeps, one for every set $V_{ij}$. During a parallel
substep, every vertex $v$ of the relevant $V_{ij}$ will do the following: if it
finds that there exists a non-zero dual tensor codeword $x_v=c_v+r_v$ (with, as
before, $\|c_v\|+\|r_v\|$ minimum among possible decompositions of $x_v$), such
that $|\hat{Z}|-|\hat{Z}+x_v|\geq |x_v|/2$, then, if there are several choices
for $x_v$ it chooses one that maximises $|x_v|$, and it updates (the local views of) $\hat{C}_j$ by $c_v$
and $\hat{R}_i$ by $r_v$, as in the sequential case, so that $\hat{Z}$ is also
updated to $\hat{Z}+x_v$. 
Every vertex $v$ of $V_{ij}$ does this procedure in parallel, after which the
decoder repeats the procedure with another set $V_{i'j'}$. When the four sets of
vertices $V_{00},V_{01},V_{10},V_{11}$ have been dealt with as described, the
round of parallel decoding is complete.

The whole procedure is summarised in Algorithm~\ref{algo:mismatch}.
Note that, for the sake of readability, the algorithm is described with a
failure detecting mechanism that is not truly part of the parallel procedure.
The fully parallel decoder will simply apply a predetermined number of decoding
rounds.

\begin{algo}{General decoding procedure}
\Input{a pair of syndromes $(s^x, s^z)$}
\Output{Either {\bf failure} or a vector $(\hat{\error}^x, \hat{\error}^z) \in \F_2^{2Q}$ with syndrome $(s^x, s^z)$.}
\label{algo:decoder}

\BlankLine
\tcc{\AL{Beginning of preprocessing phase.}}
	Set $Z^x= 0$, $Z^z = 0$, $\hat{\error}^x=0, \hat{\error}^z=0$.\;
	\For{$v \in V_1$}
	{
		Set $\eps_v$ to be the error of minimum weight with local syndrome coinciding with $s^x$.\;
		$Z^x \leftarrow Z^x + \eps_v$.\;
	}
	\For{$v \in V_0$}
	{
		Set $\eps_v$ to be the error of minimum weight with local syndrome coinciding with $s^z$.\;
		$Z^z \leftarrow Z^z + \eps_v$.\;
	}
\tcc{\AL{End of preprocessing phase.}}
	Compute $(\hat{C}_0^x, \hat{R}_0^x, \hat{C}_1^x, \hat{R}_1^x) = \text{mismatch-decomposition}(Z^x)$. \tcc{\AL{Calls Algorithm 2 or 3}} 
	Compute $(\hat{C}_0^z, \hat{R}_0^z, \hat{C}_1^z, \hat{R}_1^z) = \text{mismatch-decomposition}(Z^z)$.  \tcc{\AL{Calls Algorithm 2 or 3}} 
\tcc{\AL{Beginning of postprocessing phase.}}
	\For{$v \in V_{10}$}
	{
		Set $e^x_v = \eps_v + c_v + r_v$ where $c_v$ and $r_v$ are the local values of $\hat{C}_0^x$ and $\hat{R}_1^x$.\; 
		}
	\For{$v \in V_{00}$}
	{
		Set $e^z_v = \eps_v + c_v + r_v$ where $c_v$ and $r_v$ are the local values of $\hat{C}_{0}^z$ and $\hat{R}_{0}^z$.\; 
	}
	{\bf return $\hat{\error}^x$, $\hat{\error}^z$ corresponding to the decompositions $(e^x_v)_{v\in V_{10}}$, $(e^z_v)_{v\in V_{00}}$.}

\tcc{\AL{End of postprocessing phase.}}
\end{algo}

\begin{algo}{Sequential mismatch decomposition procedure (with parameter $\eps$)}
\Input{a mismatch $Z \in \F_2^{Q}$}
\Output{Either {\bf failure} or $(\hat{C}_0, \hat{R}_0, \hat{C}_1, \hat{R}_1)$ such that $\hat{C}_0 + \hat{R}_0 + \hat{C}_1 + \hat{R}_1 = Z$.}
\label{algo:seq-mismatch}

\BlankLine
	Set $\hat{C}_0=0, \hat{R}_0=0, \hat{C}_1=0, \hat{R}_1 = 0$ and $\hat{Z}=Z$.\;
	\While{$\hat{Z} \ne 0$}
	{
		{\bf if} $\exists v \in V_{ij}, c_v + r_v$ such that  $|\hat{Z}| - |\hat{Z} + c_v+r_v| \geq (1-\eps) |c_v+r_v|$, choosing $c_v,r_v$ such that $\|c_v\|+\|r_v\|$ is minimum among
$c_v,r_v$ such that $x_v=c_v+r_v$, \;
		\tcc{\AL{with $c_v \in C_A \otimes \F_2^B, r_v \in \F_2^A \otimes C_B$ for $\sigma_x$-type errors, and $c_v \in C_A^\perp \otimes \F_2^B, r_v \in \F_2^A \otimes C_B^\perp$ for $\sigma_z$-type errors.}}
		{\bf then} \;
		$\hat{C}_j \leftarrow \hat{C}_j + c_v$\;
		$\hat{R}_i \leftarrow \hat{R}_i + r_v$\;
		$\hat{Z} \leftarrow \hat{Z} + c_v+r_v$ \;
		{\bf else return Failure}.\;
	}
	{\bf return $(\hat{C}_0, \hat{R}_0, \hat{C}_1, \hat{R}_1)$.}
\end{algo}

\begin{algo}{Parallel mismatch decomposition procedure}
\Input{a mismatch $Z \in \F_2^{Q}$}
\Output{Either {\bf failure} or $(\hat{C}_0, \hat{R}_0, \hat{C}_1, \hat{R}_1)$ such that $\hat{C}_0 + \hat{R}_0 + \hat{C}_1 + \hat{R}_1 = Z$.}
\label{algo:mismatch}

\BlankLine
	Set $\hat{C}_0=0, \hat{R}_0=0, \hat{C}_1=0, \hat{R}_1 = 0$ and $\hat{Z}=Z$.\;
	\While{$\hat{Z} \ne 0$}
	{       
                $Temp=\hat{Z}$.\;
                \For{$(i,j)\in\{(00),(01),(10),(11)\}$}
                {
			\For{$v \in V_{ij}$,} 
			{	\If{$\exists \, x_v=c_v+ r_v\neq 0$ such that $|\hat{Z}|
-|\hat{Z}+x_v| \geq |x_v|/2$}
				{
					update $\hat{C}_{j} \leftarrow
\hat{C}_{j} +c_v$, $\hat{R}_{i} \leftarrow \hat{R}_{i} + r_v$,
$\hat{Z}\leftarrow\hat{Z}+x_v$, choosing $|x_v|$ maximum among all possible
choices, and $c_v,r_v$ being chosen such that $\|c_v\|+\|r_v\|$ is minimum among
$c_v,r_v$ such that $x_v=c_v+r_v$.\;
				}
			}

                }
		{\bf if} $\hat{Z}=Temp$ {\bf return Failure}.\;
	}

	{\bf return $(\hat{C}_0, \hat{R}_0, \hat{C}_1, \hat{R}_1)$.}
\end{algo}

\section{Sequential decoding}
\label{sub:proof}

Both the sequential decoder and the parallel decoder look for a
decomposition of the mismatch vector $Z$ as a sum, over a subset of
vertices $v$, of local dual tensor codewords $c_v+r_v$ supported by
$Q$-neighbourhoods $Q(v)$. The mismatch decomposition procedure is the core of
the decoding algorithm, both in the sequential and in the parallel case.

The sequential mismatch decomposition procedure is parameterised by a constant
$\eps \in (0,1)$
and proceeds in the natural way: it looks for a vertex $v \in V$ together with a
local dual tensor codeword $x_v = c_v+r_v$ (with $c_v \in C_A \otimes \F_2^B$,
$r_v \in \F_2^A \otimes C_B$) such that flipping $x_v$ decreases the Hamming weight of the mismatch by
at least $(1-\eps) |x_v|$, in other words, such that $|\hat{Z}+x_v|\leq |\hat{Z}|-(1-\eps)|x_v|$, where
$\hat{Z}$ is the current value of the mismatch, initiated at the original
mismatch vector $Z$. It then proceeds
by updating the current mismatch value to $\hat{Z}=\hat{Z}+x_v$, and continues
in this way until $\hat{Z}=0$, at which point it outputs the sum of the updates
$x_v$, which equals $Z$.
Theorem \ref{thm:main} below states that the required local codeword $x_v$
always exists for a given mismatch $\hat{Z}$, provided a minimum decomposition
of $\hat{Z}$ has sufficiently few active vertices. We now explain what this
means.

We may reorganise a decomposition $\hat{Z}=\sum_{v\in V} c_v+r_v$ of $\hat{Z}$ into
a sum of local
dual tensor codewords as 
\begin{equation}
\label{eq:decomp}
\hat{Z} = \hat{C}_0 + \hat{R}_0 + \hat{C}_1 +\hat{R}_1,
\end{equation}
where $\hat{C}_i=\sum_{v\in V_{ii}\cup V_{\bari i}}c_v$ and 
$\hat{R}_j=\sum_{v\in V_{jj}\cup V_{j\barj}}r_v$.
Now we may decompose $\hat{C}_0$ as a sum of of $c_v$'s, for $v$ restricted
to $V_{00}$ (as opposed to $v$ ranging over $V_{00}\cup V_{10}$), in which case
this decomposition is unique, and  
we shall
denote $\|\hat{C}_0\|$ the sum of all corresponding $\|c_v\|$s, where we recall that
$\|c_v\|$ denotes the number of individual column vectors that make up the local
codeword in $C_A\otimes \F_2^B$. 
Note that we obtain the same value $\|\hat{C}_0\|$, if we decompose
$\hat{C}_0$ as a sum of $c_v$'s, for $v$ ranging in $V_{10}$.
The quantities $\|\hat{C}_1\|$, $\|\hat{R}_0\|$, $\|\hat{R}_1\|$ are defined
similarly. We now define the {\em norm} $\|\hat{Z}\|$ of $\hat{Z}$ to be
the minimum value of 
\begin{equation}
\label{eq:norm}
\|\hat{C}_0\| + \|\hat{R}_0\| + \|\hat{C}_1\| +\|\hat{R}_1\|
\end{equation}
over all possible decompositions \eqref{eq:decomp} of $\hat{Z}$. We shall say
that a decomposition \eqref{eq:decomp} of $\hat{Z}$ is minimum if
\eqref{eq:norm} equals $\|\hat{Z}\|$.
Finally, we shall say that for a decomposition \eqref{eq:decomp}, the
vertex $v$ in $V_{ij}$ is {\em active}, if $\hat{C}_j+\hat{R}_i$ is non-zero on $Q(v)$.
We may now state Theorem~\ref{thm:main}, which is the core technical result
for the decoder analysis:
\begin{theo}
\label{thm:main}
Fix $\eps \in (0,1)$. 
If, for $i,j\in\{0,1\}$, the sets of active vertices $S_{ij}\subset V_{ij}$ for a minimum decomposition of a
mismatch vector $Z$ satisfy
\[
|S_{ij}|\leq \frac{1}{2^{12}}\delta^2\eps^3\kappa |V_{00}|
\]
then there exists some vertex $v$ and some codeword $x_v$ of the dual tensor code such that 
\[ |Z| - |Z + x_v| \geq (1-\eps) |x_v|.\]
\end{theo}

From Theorem~\ref{thm:main}, will follow Theorem~\ref{thm:sequential},
which states that if the weight $|\error|$ of the
original error $\error$ is sufficiently small, then the number of active
vertices of a minimum decomposition of $\hat{Z}$ will satisfy the hypothesis of
Theorem~\ref{thm:main} throughout the sequential decoding procedure, meaning
that the sequential mismatch decomposition procedure will always converge.

\begin{theo}
\label{thm:sequential}
Fix $\eps\in (0,1)$.
If the Hamming weight $|\error|$ is less than 
\[
\frac{1}{2^{11}}\min\Big(\frac{\eps^3}{16}, \kappa \Big)(1-\eps)\delta^2\kappa^2 \frac{n}{\Delta},
\]
then the sequential decoder with parameter $\eps$ returns a valid correction $\hat{\error}$, namely a correction equivalent to $\error$ (differing from $\error$ by an element of the stabilizer group).
\end{theo}

\subsection{Proof of Theorem \ref{thm:main}}
\label{sub:thm7}

The proof strategy follows the same footsteps as the proof of the minimum
distance. When studying the minimum distance, we analysed the set $C_j + R_i$,
which coincided with $C_{\overline{\jmath}} + R_{\overline{\imath}}$. 
For the decoding, we start by identifying the mismatch $Z$ associated with the
error and take a minimum decomposition $C_0 + R_0+C_1+R_1$. The relevant set is
now the support of $(C_j + R_i)\cap(C_{\barj} \cap R_{\bari})$, which is smaller than the set considered when analysing the minimum distance. 
However, under the assumption that the sequential decoder is stalled, this set cannot be too small, and essentially the same techniques as before will allow us to arrive at a contradiction.

\subsubsection{A stalled sequential decoder, Exceptional vertices, ordinary rows and columns}

We consider a $\sigma_x$-type error $\error$ 
and define its associated mismatch $Z$. 
We work with a minimal decomposition of $Z$:
\[ Z = C_0 + R_0 + C_1+R_1,\]
meaning that the quantity $\|C_0\| + \|R_0\| + \|C_1\|+\|R_1\|$ is minimal. 
To each vertex $v \in V$, this decomposition associates codewords $c_v \in C_A \otimes \F_2^B$ and $r_v \in \F_2^A \otimes C_B$. 
We say that a vertex $v\in V_{ij}$ is an \emph{active vertex} if $c_v + r_v \ne 0$, \textit{i.e.}\
if $C_j+R_i$ is non-zero on $Q(v)$, and we
denote by $S_{ij}$ the sets of active vertices in $V_{ij}$. 

The sequential decoder with parameter $\eps$ searches for some vertex $v \in V$
and a dual tensor codeword $x_v \in (C_A^\perp \otimes C_B^\perp)^\perp = C_A
\otimes \F_2^B + \F_2^A \otimes C_B$ such that flipping $x_v$ decrease the
Hamming weight of $Z$ by at least $(1-\eps) |x_v|$. 
To prove Theorem \ref{thm:main} we will assume there is no such vertex and work
towards a contradiction.

We will follow the blueprint of Section~\ref{sec:distance}, and define
exceptional and ordinary vertices of $S_{ij}$ as before, namely a vertex 
$v\in S_{ij}$ is said to be exceptional, if the local dual tensor codeword
$x_v=c_v+r_v$,
equal to the restriction of $C_j+R_i$ to $Q(v)$, satisfies
$\|c_v\|+\|r_v\|\geq\alpha\Delta$.
Here we will take 
$\alpha=\frac{1}{2^{10}}\delta^2\eps^2$.
The set of exceptional vertices of $S_{ij}$ is denoted by $S_{ij}^e$ and
non-exceptional vertices are called {\em ordinary}. 
Let us furthermore call an ordinary column (row) of $v\in S_{ij}$ a column (row) of the
$Q$-neighbourhood $Q(v)$ on which $C_j$ ($R_i$) is non-zero, and for a vertex
$v$ that is ordinary. Note that when talking about ordinary columns (or rows) it is
important to specify for which vertex $v$, since this column appears in two
different local views for two different vertices, and may be ordinary for one
vertex and not for the other.

\begin{lemma} \label{lem:stalled} Assume that the sequential decoder of parameter $\eps$ is stalled.
For all $v \in S_{ij}$, and
all dual tensor codewords $x_v$, components of $C_j+R_i$, we have
\[ |x_v \cap (C_{\barj}+R_{\bari})| \geq \frac{\eps}{2} |x_v|.\]
Furthermore, let $y_v$ be the subvector of $C_j$ supported
by some ordinary column for $v\in S_{ij}$. Then 
\[|y_v\cap (C_{\barj}+R_{\bari})|\geq
\frac{\eps}{4}|y_v|.
\]
\end{lemma}

\begin{proof}
Note that for any two binary vectors $x,z$, identifying them with their supports
we have that $|z|-|z+x|\leq (1-\eps)|x|$ is equivalent to $2|z\cap x|\leq
(2-\eps)|x|$, since $|z+x|=|z|-2|z\cap x|+|x|$.
Since the decoder is stalled, we have 
\[ |Z| - |Z + x_v| \leq (1-\eps) |x_v|\]
which therefore gives 
\[2 |Z \cap x_v|\leq (2-\eps) |x_v|.\]
Note that $x_v \subset C_j+R_i$ 
and therefore $Z
\cap x_v = x_v + ((C_{\barj}+R_{\bari})\cap x_v)$ and $|Z\cap x_v| = |x_v| -
|(C_{\barj}+R_{\bari}) \cap x_v|$. Combining this with the previous inequality proves
the first claim of the Lemma.

To prove the second claim of the Lemma we argue that, since $y_v$ is an ordinary column
vector,  that there is some
$y'\subset y_v$ such that $|y'|\leq \alpha\Delta$ and $y_v+y'\subset C_j+R_i$.
Specifically, $y'$ is supported by the intersection of $y_v$ and $R_i$.
From our choice of $\alpha$ we clearly have $\alpha\Delta\leq
\frac{\eps}{2}\delta\Delta$, so that $|y'|\leq \frac{\eps}{4}|y_v|$.
We have
\[
2 |Z \cap y_v|\leq (2-\eps) |y_v|,
\]
otherwise we could decode at vertex $v$ by flipping $y_v$.
Since $y_v+y'\subset y_v$, we can write $2 |Z \cap (y_v+y')|\leq (2-\eps) |y_v|$,
and since $y_v+y'\subset C_j+R_i$, we have $|Z\cap (y_v+y')| = |y_v+y'| -
|(C_{\barj}+R_{\bari}) \cap (y_v+y')|$, whence
\[
2|y_v|-2|y'| -(2-\eps)|y_v| \leq 2|(C_{\barj}+R_{\bari}) \cap (y_v+y')|\leq
2|(C_{\barj}+R_{\bari}) \cap y_v|
\]
which proves the claim, since $|y'|\leq \frac{\eps}{4}|y_v|$.
\end{proof}

Recall that minimality of the representation $(C_0,R_0,C_1,R_1)$ and the robustness property (Lemma \ref{thm:robust}) imply that
\[ |c_v + r_v| \geq \kappa (\|c_v\| + \|r_v\|) \Delta\]
whenever $c_v+r_v$ is the local representation of $C_j+R_i$ at $v\in S_{ij}$. 
In particular, an exceptional vertex is such that $|c_v + r_v| \geq \alpha \kappa \Delta^2$. As before we denote by $S_{ij}^e$ the set of exceptional vertices in $S_{ij}$.

\subsubsection{Exceptional vertices are rare}

\begin{lemma}\label{lem:Se} 
Let $i=0,1$. Under the hypothesis $|S_{ij}|\leq \frac{\alpha\kappa\eps}{4}|V_{00}|$, we have 
\[
|S_{ii}^e|\leq \frac{2^{8}}{\alpha^2 \kappa^2 \eps^2}\frac{1}{\Delta^{2}}|S_{0}|
\quad\text{and}\quad |S_{i\bari}^e|\leq  \frac{2^{8}}{\alpha^2 \kappa^2 \eps^2}\frac{1}{\Delta^{2}}|S_{1}|.
\]
\end{lemma}

\begin{proof}
The proof follows closely that of Lemma~\ref{lem:excep}.
We prove the upper bound for $S_{00}^e$, the other cases being similar. 
Viewing $(C_0+R_0)\cap(C_1+R_1)$ as a subgraph of $\G_0^\square$, we have that vertices of
$S_{00}^e$ have degree at least $\kappa\alpha\eps\Delta^2/2$, by
Lemma~\ref{lem:stalled} and robustness.

By the Expander mixing Lemma in $\G_0^\square$,
we therefore have
\[
\frac{1}{2} \alpha \kappa \eps \Delta^2|S_{00}^e|\leq |E(S_{00}^e,S_{11})| \leq \Delta^2\frac{|S_{00}^e||S_{11}|}{|V_{00}|} +
4\Delta\sqrt{|S_{00}^e||S_{11}|}.
\]
Writing $|S_{11}|/|V_{00}|\leq \frac{\alpha\kappa\eps}{4}$ by assumption, 
we get
\begin{align*}
\frac{1}{4} \alpha \kappa \eps \Delta^2|S_{00}^e|^{1/2} &\leq
4\Delta|S_{11}|^{1/2}\\
|S_{00}^e| &\leq \frac{2^{8}}{\alpha^2 \kappa^2 \eps^2}\frac{1}{\Delta^{2}}|S_{11}|.
\end{align*}
which proves the claimed result, since $|S_{11}|\leq |S_0|$.
\end{proof}

\subsubsection{Ordinary columns and rows cluster}
\begin{lemma}
\label{lem:T2}
Let $T$ be the set of vertices $v$ of $V_{ij}$ whose $Q$-neighbourhood $Q(v)$
contains a column that is a ordinary column for the neighbouring vertex $w\in
V_{\bari j}$ whose local view shares this very column.
If we have $|S_{i\barj}|\leq
\delta\eps |V_{i\barj}|/8$, then 
\[
|T|\leq \frac{256}{\delta^2\eps^2\Delta}|S_{i\barj}|.
\]
Exactly the same result holds for the set $T\subset S_{\bari\barj}$ of vertices on whose $Q$-neighbourhood
$R_{\bari}$ displays a non-zero row codeword of $C_B$ that is shared by the local view
of an ordinary vertex $w\in S_{\bari j}$. 
\end{lemma}

\begin{proof}
We follow closely the proof of Lemma~\ref{lem:T}. We deal with the case when $i=j=0$ and $T$ is defined as the set of vertices
of $V_{00}$ whose $Q$-neighbourhoods share a non-zero column vector with an
ordinary vertex of $V_{10}$. The other cases will hold by symmetry.
 
Again we keep only the two sets of vertices $V_{00}$ and $V_{01}$, and collapse
rows of local views to single edges, and we look at the graph induced by the
squares of $(C_1+ R_0)\cap (C_0+R_1)$ on the vertex set $T\cup S_{01}$. What the
second claim of Lemma~\ref{lem:stalled} tells us, is that the degree of any vertex of $T$ in
this subgraph is at least $\delta\Delta\frac{\eps}{4}$.

As in the proof of Lemma~\ref{lem:T}, we apply the expander mixing Lemma in the
double cover of $\Cay(G,A)$ to obtain,
\[
|T|\frac{\delta \eps \Delta}{4}\leq | E(S_{01},T)|\leq \Delta\frac{|S_{01}||T|}{|V_{01}|} +
2\sqrt{\Delta}\sqrt{|S_{01}||T|}.
\]
Applying the hypothesis $|S_{01}|\leq\delta\eps |V_{01}|/8$ we get
\[
|T|\frac{\delta \eps\Delta}{8}\leq 2\sqrt{\Delta}\sqrt{|S_{01}||T|}
\]
and the result follows.
\end{proof}

\subsubsection{Obtaining a contradiction}

\begin{proof}[Proof of Theorem~\ref{thm:main}]
The hypothesis of the theorem translates into
$|S_{ij}|\leq \frac{\alpha\eps\kappa}{4}|V_{00}|$, since we defined $\alpha = \delta^2 \eps^2/2^{10}$.
Therefore Lemma~\ref{lem:Se} holds, and so does 
Lemma~\ref{lem:T2}, since clearly $\frac 14\alpha\eps\kappa \leq \delta\eps/8$.

Without loss of generality, let us suppose $|S_1|\geq |S_0|$ (otherwise invert
their roles in the argument below) and $|S_{10}|\geq |S_{01}|$ (otherwise invert
their roles).

Because there are so few exceptional vertices in $S_{10}$ (by
Lemma~\ref{lem:Se}), the number
of ordinary vertices of $S_{10}$ is almost equal to $|S_{10}|$ namely it is at
least $a|S_{10}|$ for any constant $a<1$ and $\Delta$ large enough.

We obtain that either the number of ordinary rows or the number of ordinary
columns in (the $Q$-neighbourhoods of vertices of) $S_{10}$ is at least $a|S_{10}|/2$. 
Suppose without loss of generality that the number of ordinary columns is at
least $a|S_{10}|/2$.
Applying Lemma~\ref{lem:T2}, they must cluster among the $Q$-neighbourhoods of a
set $T\subset S_{00}$ of size
at most $\frac{256}{\delta^2\eps^2\Delta}|S_{01}|\leq
\frac{256}{\delta^2\eps^2\Delta}|S_{10}|$ (since we have supposed $|S_{01}|\leq
|S_{10}|$).
Therefore, the average number of non-zero columns of $C_0$ on the
$Q$-neighbourhoods of the vertices of this set
$T$ is at least $a\frac{\delta^2\eps^2\Delta}{512}$, which is close to twice the
minimum norm $\alpha\Delta$ of a local view of $R_i+C_j$ for an exceptional vertex.

Therefore, a constant proportion of vertices of $T$ must be exceptional vertices
of $S_{00}$. Now since the $Q$-neighbourhood of a vertex can host at most
$\Delta$ columns, we also have $|T|\geq \frac 1\Delta \frac a2|S_{10}|\geq \frac
a{4\Delta} |S_1|$ (since $|S_{10}|\geq |S_{01}|$), and since $|S_1|\geq |S_0|$
we have $|T|\geq \frac{a}{4\Delta}|S_0|$. Therefore, for large enough $\Delta$,
we have a contradiction with Lemma~\ref{lem:Se} which limits the number of
exceptional vertices of $S_{00}$ to not more than $O(1/\Delta^2)|S_0|$.
\end{proof}

\subsection{Proof of Theorem \ref{thm:sequential}}
\label{sub:thm8}
Without loss of generality we may suppose the error $\error$ to be a $\sigma_x$-type error. 

We need to guarantee that the upper bound on $|S_{ij}|$ required by
Theorem~\ref{thm:main} is satisfied throughout the decoding procedure, until we
reach a zero mismatch $\hat{Z}$. Recall that $S_{ij}$ is the set of active
vertices of $V_{ij}$
in a minimum decomposition of $\hat{Z}$. We will argue that $|S_{ij}|\leq
\|\hat{Z}\|$, so we track the evolution of $\|\hat{Z}\|$ during the mismatch
decomposition procedure.

\begin{lemma}
\label{lem:preproc}
If $Z$ is the mismatch vector that the preprocessing phase associates to an
error vector $\error$, then we have
\begin{equation}
\label{eqn:bound-Z}
|Z|\leq 4|\error|
\end{equation}
and 
\[
\|Z\|\leq \frac{4}{\kappa\Delta}|\error|.
\]
\end{lemma}

\begin{proof}
During the preprocessing phase, we have that
$Z=\sum_{v\in V_1}\eps_v=\sum_{v\in V_1}x_v$, 
where $e_v = \eps_v + x_v$ is the
projection of the error vector $\error$ on $Q(v)$ and where $x_v=c_v+r_v$ is the dual
tensor codeword that is the difference between $e_v$ and the decoder's initial
evaluation $\eps_v$ of the error. 
The minimality of $|\eps_v| = |e_v+ x_v|$ implies that $|\eps_v| \leq |e_v|$ and therefore that $|x_v| \leq 2 |e_v|$.
This gives
\[
|Z| \leq \sum_{v \in V_1} |x_v| \leq 2 \sum_{v \in V_{10}} |e_v| = 4
|\error|
\]
which proves the first point.
Writing $\|x_v\|:=\|c_v\|+\|r_v\|$ and applying robustness to $x_v$, we also have 
$\kappa\Delta\|x_v\|\leq |x_v|\leq 2|e_v|$, whence
\[
\|x_v\| \leq\frac{2}{\kappa\Delta}|e_v|.
\]
Summing over all vertices of $V_1$, we obtain
\[
\|Z\|\leq\sum_{v\in V_1}\|x_v\|\leq\sum_{v\in V_{01}}\|x_v\| + \sum_{v\in V_{10}}\|x_v\|
\leq \frac{2}{\kappa\Delta}\sum_{v\in V_{01}}|e_v| +
\frac{2}{\kappa\Delta}\sum_{v\in V_{10}}|e_v| =
\frac{4}{\kappa\Delta}|\error|
\]
since $|\error| = \sum_{v\in V_{01}}|e_v| =\sum_{v\in V_{10}}|e_v|$.
\end{proof}

Assume that after the $m^{\text{th}}$ round of sequential decoding, the decoder
has flipped successively the local dual tensor codewords $x_{1}, x_{2}, \ldots, x_{m}$.
We must have in particular
\[
(1-\eps)(|x_{1}|+|x_{2}|+\cdots |x_{m}|)\leq |Z|.
\]
Now, every time we modify $\hat{Z}$ by adding some $x_i$ to it, we have that
$\|\hat{Z}\|$ is increased by at most $\|x_i\|$. Robustness implies that
$\|x_i\|\leq \frac{1}{\kappa\Delta}|x_i|$ and we may therefore write
\begin{equation}
\label{eq:sumofnorms}
\|x_{1}\|+\|x_{2}\|+\cdots \|x_{m}\| \leq \frac{1}{\kappa\Delta(1-\eps)}|Z|.
\end{equation}
Applying Lemma~\ref{lem:preproc}
we therefore obtain that
\[
\|\hat{Z}\|\leq \|Z\| + \|x_{1}\|+\|x_{2}\|+\cdots \|x_{m}\| \leq
\frac{4}{\kappa\Delta}\left(1+\frac{1}{1-\eps}\right)|\error|\leq
\frac{8}{(1-\eps)\kappa\Delta}|\error|
\]
since $1+1/(1-\eps)\leq 2/(1-\eps)$. Therefore, if we impose the
condition
\[
|\error|\leq \frac{(1-\eps)\kappa\Delta}{8}\frac{1}{2^{12}}\delta^2\eps^3\kappa|V_{00}| =
\frac{1}{2^{11}}\frac{\eps^3}{16}\kappa^2\delta^2(1-\eps)\frac{n}{\Delta}
\]
we obtain that $|S_{ij}|\leq\|\hat{Z}\|$ must always be bounded from above by
$\frac{1}{2^{12}}\delta^2\eps^3\kappa|V_{00}|$ throughout the decoding
procedure, so that Theorem~\ref{thm:main} always applies and decoding can always
continue.

We also need to check that the output $\hat{\error}$ of the decoder is correct.
This will be the case provided that $|\error + \hat{\error}| < d_{\min} (\eQ)$.
Recall that 
\begin{equation}
\label{eq:e+ehat}
|\error + \hat{\error}| = |\sum_{v\in V_{10}}(e_v+\eps_v) + \hat{C}_0 + \hat{R}_1|
\end{equation}
and that $|\sum_{v\in V_{10}}(e_v+\eps_v)|=|\sum_{v\in V_{10}}x_v| \leq 2 |\error|$.
To evaluate $|\hat{C}_0 +\hat{R}_1|$, we make the remark that every local
dual tensor codeword $x_i$ that is used by an iteration of the sequential
decoder 
contributes at most $\|x_i\|$ non-zero row and column vectors to
$\hat{C}_0 +\hat{R}_1$. Therefore, 
\[
|\hat{C}_0+\hat{R}_1|\leq \Delta\sum_i\|x_i\|
\]
where the sum runs over all local increments used by the decoder. Applying
\eqref{eq:sumofnorms}, we obtain
\[
|\hat{C}_0+\hat{R}_1|\leq \frac{1}{\kappa(1-\eps)}|Z|.
\]
Writing $|Z|\leq 4|\error|$ (from \eqref{eqn:bound-Z}), we get 
\[
|\hat{C}_0+\hat{R}_1|\leq \frac {4}{\kappa(1-\eps)}|\error|.
\]
Since $\kappa <1$, we may write $2|\error|\leq \frac{4}{\kappa(1-\eps)}|\error|$, and
\eqref{eq:e+ehat} now gives us
\[
|\error + \hat{\error}|\leq \frac{4}{\kappa(1-\eps)}|\error|+|\hat{C}_0 +
\hat{R}_1|\leq\frac{8}{\kappa(1-\eps)}|\error|
\]
which is smaller than the minimum distance of the quantum code whenever
$|\error|\leq \frac{\kappa}{2^{11}}\kappa^2\delta^2 (1-\eps)\frac{n}{\Delta}$.
This concludes the proof of Theorem \ref{thm:sequential}.

\section{Parallel Decoding}
Without loss of generality, we again assume the error $\error$ to be a $\sigma_x$-type error.
The analysis of the parallel decoder rests upon the following lemma.

\begin{lemma}
\label{lem:oneround}
Let $\eps\in(0,1/6)$, and assume that the state of the current mismatch $Z$ is
such that the sequential decoder of parameter $\eps$
converges. Then after one round of parallel decoding, the weight of the mismatch
has been reduced by at least 
$\frac 14(\frac 14-\frac 32\eps)|Z|$.
\end{lemma}

To obtain that the parallel decoder converges in a logarithmic number of steps,
it will remain to show that, provided the initial error vector is of
sufficiently small weight, the mismatch will remain in a decodable state by the
sequential decoder of parameter $\eps$ after an arbitrary number of parallel
decoding steps. The result will then follow from applying
Lemma~\ref{lem:oneround} iteratively.

\begin{proof}[Proof of Lemma~\ref{lem:oneround}]
Let $x_1,x_2,\ldots,x_i\ldots x_m$ be a sequence of local non-zero dual tensor codewords such
that, for every $i\geq 1$,
\begin{equation}
\label{eq:sequential}
|Z_{i-1}|-|Z_{i-1}+x_i|\geq (1-\eps)|x_i|
\end{equation}
where $Z_0=Z$ and
$Z_i=Z_{i-1}+x_i$. That the sequential decoder of parameter $\eps$ converges
means that whenever $x_1,\ldots ,x_i$ satisfy \eqref{eq:sequential} and
$Z_i\neq\varnothing$, then the sequence $x_1,\ldots ,x_i$ can be augmented by some
$x_{i+1}$ satisfying \eqref{eq:sequential}, until eventually $Z_m=\varnothing$.

It will be useful to keep in mind that the condition \eqref{eq:sequential} is
equivalent to the condition $|Z_{i-1}\cap x_i|\geq (1-\eps/2)|x_i|$.

Let us define $x_1'=x_1\cap Z$ and for $i\geq 2$,
$x_i'=x_i\cap (Z\setminus\bigcup_{j<i}x_j)$, so that the $x_i'$
are disjoint and partition $Z$. It may happen that some $x_i'$ are empty.

Let $G$ (good) be the subset of indices $1\leq i \leq m$ for which it holds that
\begin{equation}\label{eq:zi}
\left|x_i'\right| \geq \left(1-\frac
32\eps\right)|x_i|
\end{equation}
and let $B$ (bad) be the set of remaining indices. Let us first prove:

\noindent
\textbf{Claim 1.} We have
\[
|Z\cap \bigcup_{i\in G}x_i'| \geq |Z|/2.
\] 
Let us write the mismatch vector after addition of $x_i$ as $Z_i=Z_i'\cup A_i$,
where the union is disjoint and where we have set
\[
Z_i' = Z\setminus\bigcup_{j\leq i}x_i'.
\]
The situation is illustrated on Figure~\ref{fig:Z}.

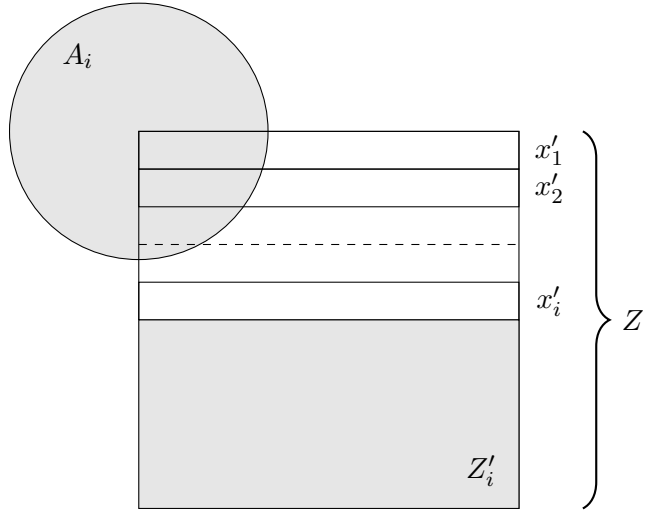
\begin{figure}[h]
\begin{center}
\begin{tikzpicture}
\filldraw[fill=black!10] (0,5) circle (1.7cm);
\draw (0,0) rectangle (5,5);
\draw (0,4.5) rectangle (5,5);
\draw (0,4) rectangle (5,4.5);
\draw[style=dashed] (0,3.5) -- (5,3.5);
\draw (0,2.5) rectangle (5,3);
\node (z1) at (5.4,4.75) {$x_1'$};
\node (z2) at (5.4,4.25) {$x_2'$};
\node (zi) at (5.4,2.75) {$x_i'$};
\node (Ai) at (-0.8,6) {$A_i$};
\filldraw[fill=black!10] (0,0) rectangle (5,2.5) {};
\node (Z) at (4.5,0.5) {$Z_i'$};
\draw [thick,decorate,decoration={brace,amplitude=10pt,mirror,raise=4pt},yshift=0pt]
(5.7,0) -- (5.7,5) node (x) [midway,xshift=8mm] {$Z$};
\end{tikzpicture}
\end{center}
\caption{The state $Z_i=Z_i'\cup A_i$ of the mismatch after $i$ steps (shaded): The supports of
$x_1',\ldots x_i'$ have been removed from the original mismatch $Z$, and an
additional set $A_i$ has been added.}
\label{fig:Z}
\end{figure}

Whenever $x_i$ is added to the current mismatch $Z_{i-1}$, $x_i'$ is removed
from the support of $Z_i'$, and the set of 'additional' bits/coordinates $A_{i-1}$ 
is modified to become $A_i$. We remark that for every $i$, the number of bits that are added to 
$A_{i-1}$, \textit{i.e.}\ $|A_i\setminus A_{i-1}|$, is at most
$\frac{\eps}{2}|x_i|$, since $x_i$ is an $\eps$-sequential decoder increment. 
Furthermore, for $i\in B$, the condition \eqref{eq:zi} implies that 
\[
|A_{i-1}\setminus A_i|\geq \eps|x_i|
\]
so that $|A_{i-1}|-|A_i| \geq \frac{\eps}{2}|x_i|$. Summarising, we have
\begin{align*}
|A_i|-|A_{i-1}| &\leq \phantom{-}\frac{\eps}{2}|x_i|\quad\text{for $i\in G$,}\\
|A_i|-|A_{i-1}| &\leq -\frac{\eps}{2}|x_i|\quad\text{for $i\in B$,}
\end{align*}
from which we obtain 
\[0 = |A_m| - |A_0| = \sum_{i=1}^m |A_i | - |A_{i-1}| \leq \frac{\eps}{2}
\Big( \sum_{i \in G} |x_i| - \sum_{i \in B} |x_i| \Big),\]
that is,
\[
\sum_{i\in B}|x_i|\leq\sum_{i\in G}|x_i|.
\]
From the definitions of the sets $B$ and $G$, we have therefore
\[
\sum_{i\in B}|x_i'|\leq \sum_{i\in B}(1-\frac 32\eps)|x_i|\leq \sum_{i\in G}(1-\frac
32\eps)|x_i|\leq \sum_{i\in G}|x_i'|
\]
which proves the claim since the $x_i'$ partition $Z$.

To finish the proof of the lemma, consider the vector $u$, equal to the {\em union} of all the bits that have been
flipped by the decoder after all four substeps of the parallel decoding step. Consider a local dual tensor codeword
$x_i$ from the above sequence $x_1,\ldots ,x_m$, for $i\in G$. We wish to prove that a constant fraction of its support is included in the support of $u$. The local vector
$x_i$ sits in the local view of some vertex $v$. 
Let $y=|u\cap
x_i'|$. We must have $|x_i'|-y<\frac 34|x_i|$. Indeed, suppose otherwise: let
$z_v$ be the local dual tensor codeword that the parallel local decoder at $v$ has
flipped when it was its turn to decode (with possibly $z_v=0$). Note that we
must have $z_v\subset u$ by definition of $u$. We have
\begin{equation}
\label{eq:yv}
|\hat{Z}|-|\hat{Z}+z_v|\geq |z_v|/2
\end{equation}
where $\hat{Z}$ is the local view of the
mismatch at $v$, at the moment when it is $v$'s turn to decode. Now since
$x_i'\setminus u$ has been left untouched by the decoder throughout the round of
parallel decoding, we must have $(x_i'\setminus u)\subset \hat{Z}$ and
$(x_i'\setminus u)\subset \hat{Z}+z_v$. Therefore,
it must be the case that if, just after $v$ has flipped $z_v$, we flip $x_i$,
we remove at least $|x_i'|-y\geq\frac 34|x_i|$ coordinates from the current
mismatch $\hat{Z}+z_v$
and therefore add at most $|x_i|/4$ coordinates to it. In other words,
\[
|\hat{Z}+z_v| -|\hat{Z}+z_v+x_i|\geq |x_i|/2.
\]
Adding to this inequality the inequality \eqref{eq:yv}, we get
\[
|\hat{Z}| -|\hat{Z}+z_v+x_i|\geq |x_i|/2+|z_v|/2\geq |x_i+z_v|/2.
\]
Since we have at least $|x_i'\setminus u|\geq\frac 34|x_i|$ coordinates of $x_i$
that are disjoint from $u$ and hence disjoint from $z_v$, we must have
$|x_i+z_v| > |z_v|$, which means that the local decoder at $v$ would have
preferred $x_i+z_v$ over $z_v$ when it was its turn to decode, since it is
instructed to maximise the Hamming weight of the local codeword it flips.
This establishes the contradiction and proves therefore
\[
|x_i'|-y<\frac 34|x_i|.
\]
This implies 
\[y>|x_i'|-\frac 34|x_i|\geq(1-\frac 32\eps)|x_i|-\frac 34|x_i|= (\frac 14-\frac
32\eps)|x_i|
\]
since $i\in G$,
and therefore, since the
$x_i'$ are disjoint,
\[
|u| \geq \left(\frac 14-\frac 32\eps\right)\sum_{i\in G}|x_i|.
\]
By writing 
$
\sum_{i\in G}|x_i|\geq \sum_{i\in G}|x_i'|
$
and applying Claim 1, we get that
\[
|u|\geq \left(\frac 14-\frac 32\eps\right)\frac 12|Z|.
\]
We conclude the proof by noticing that if the total number of bits that are
flipped by the decoder is $N$ (where a bit may be flipped multiple times,
everytime contributing to $N$), then the mismatch weight is decreased by at
least $N/2$,
by definition of the decoding criterion; so the claim follows since the total
number of bits flipped by the decoder is at least $|u|$. 
\end{proof}

\begin{lemma}
\label{lem:simu}
Let $\eps\in (0,1/2)$.
If the Hamming weight $|\error|$ is less than
\[
\frac{1}{2^{16}}\eps^3\kappa^2\delta^2\frac{n}{\Delta},
\]
Then {\em sequential} decoding with parameter $1/2$ is guaranteed to terminate.
Furthermore, at any stage during the mismatch decoding procedure, if the
sequential decoding parameter is switched from $1/2$ to $\eps$, then mismatch
decomposition is also guaranteed to terminate.
\end{lemma}

What the conclusion of Lemma~\ref{lem:simu} means is that, if $Z=\hat{Z}_0$ is the original mismatch vector
after preprocessing, denoting by $\hat{Z}_\ell$ the mismatch vector after $\ell$
sequential decoding steps, \textit{i.e.}
\[
\hat{Z}_\ell=Z+x_1+\ldots + x_\ell,
\]
where $x_1,\ldots ,x_\ell$ is the sequence of local dual tensor codeword
increments, satisfying
\[
|\hat{Z_j}|-|\hat{Z_{j}}+x_{j+1}|\geq |x_{j+1}|/2
\]
for $j<\ell$, then it is possible to extend the sequences $x_1,\ldots ,x_\ell$ and
$\hat{Z}_0, \ldots, \hat{Z}_\ell$ to $x_1,\ldots ,x_i,\ldots x_m$ and
$\hat{Z}_0,\ldots, \hat{Z}_i, \ldots, \hat{Z}_m=0$, with 
\[
|\hat{Z_j}|-|\hat{Z_{j}}+x_{j+1}|\geq (1-\eps)|x_{j+1}|
\]
for $j=\ell\ldots m-1$.

\begin{proof}[Proof of Lemma~\ref{lem:simu}]
Assume that after a certain number of sequential decoding steps, the decoder
has flipped successively the local dual tensor codewords $x_{1}, x_{2}, \ldots,
x_{\ell}$.
Since for every codeword $x_j$ that is flipped the mismatch decreases by at
least
$|x_j|/2$,
we must have
\[
\frac 12(|x_{1}|+|x_{2}|+\cdots |x_{\ell}|)\leq |Z|.
\]
where $|Z|$ is the original mismatch vector computed after the preprocessing
phase.
Robustness implies that
$\|x_j\|\leq \frac{1}{\kappa\Delta}|x_j|$ and we may therefore write
\begin{equation}
\label{eq:sumofnorms2}
\|x_{1}\|+\|x_{2}\|+\cdots \|x_{\ell}\| \leq \frac{2}{\kappa\Delta}|Z|.
\end{equation}
Applying Lemma~\ref{lem:preproc}
we therefore obtain that
\[
\|\hat{Z}_\ell\|\leq \|Z\| + \|x_{1}\|+\|x_{2}\|+\cdots \|x_{\ell}\| \leq
\frac{4}{\kappa\Delta}\left(1+2\right)|\error|\leq
\frac{12}{\kappa\Delta}|\error|
\]
where $\hat{Z}_\ell$ is the current mismatch after $x_1,\ldots ,x_\ell$ have been
flipped, and we used that $\kappa < 1$.
Therefore, if we impose the
condition
\[
|\error|\leq \frac{1}{2^{16}}\eps^3\kappa^2\delta^2\frac{n}{\Delta}<
\frac{1}{2^{12}}\frac{\eps^3}{12}\kappa^2\delta^2\frac{n}{\Delta}=\frac{\kappa\Delta}{12}\frac{1}{2^{12}}\delta^2\eps^3\kappa|V_{00}|, 
\]
we obtain that
$\|\hat{Z}_\ell\|$ must be bounded from above by
$\frac{1}{2^{12}}\delta^2\eps^3\kappa|V_{00}|$. Since the set $S$ of active
vertices for $\|\hat{Z}_\ell\|$ in any one of the four types of vertices $V_{00},V_{01},V_{10},V_{11}$ 
must satisfy $|S|\leq |\hat{Z}_\ell |$, we have that Theorem~\ref{thm:main} applies
and there exists $x_{\ell+1}$ such that
$|\hat{Z}_\ell|-|\hat{Z}_\ell+x_{\ell+1}|\geq(1-\eps)|x_{\ell+1}|$.
Since $\eps \leq 1/2$,
we have that $x_{\ell+1}$ is also a valid choice for the
sequential decoder with parameter $1/2$, and we may therefore reapply the very
same argument iteratively to obtain the required sequence $x_{\ell+1},\ldots
,x_m$.
\end{proof}
\begin{theo}
\label{corol:parallel}
Fix $\eps \in(0,1/6)$.
If the Hamming weight $|\error|$ is less than
\[
\frac{1}{2^{12}}\min\Big(\frac{\eps^3}{16}, \kappa \Big)\delta^2\kappa^2\frac{n}{\Delta},
\]
then the parallel decoder returns a valid correction in logarithmic time. 
\end{theo}

\begin{proof}
We make the remark that a substep of parallel decoding has the same effect on
the mimatch as that of a sequence of successive increments by the sequential
decoder of parameter $1/2$. Therefore, after any number of rounds of parallel
decoding, we are in the same situation as if a sequential decoder of parameter
$1/2$ had been applied, and Lemma~\ref{lem:simu} applies, which in turn implies
that Lemma~\ref{lem:oneround} applies, so that the parallel decoder terminates
in a logarithmic number of rounds.

It remains to prove that the output $\hat{\error}$ of the decoder is such that
$|\error+\hat{\error}|$ is smaller than the minimum distance of the quantum
code. This argument is very close to that of the end of the proof of
Theorem~\ref{thm:sequential}.

Recall that 
\begin{equation}
|\error + \hat{\error}| = |\sum_{v\in V_{10}}(e_v+\eps_v) + \hat{C}_0 + \hat{R}_1|
\end{equation}
where $e_v$ is the initial error vector restricted to the local view at $v$, and
$\hat{C}_0, \hat{R}_1$ are the output of the decoder. We have $|\eps_v|\leq
|e_v|$ so that
\[
|\error + \hat{\error}| \leq 2|\error|+|\hat{C}_0 + \hat{R}_1|.
\]
To evaluate $|\hat{C}_0 +\hat{R}_1|$, we make the remark that every local
dual tensor codeword $x_i$ that is used by an iteration of the parallel 
decoder
contributes at most $\|x_i\|$ non-zero row and column vectors to
$\hat{C}_0 +\hat{R}_1$. Therefore, 
\[
|\hat{C}_0+\hat{R}_1|\leq \Delta\sum_i\|x_i\|
\]
where the sum runs over all increments applied by the decoder. Applying
robustness $\|x_i\|\leq\frac{1}{\kappa\Delta}|x_i|$ and
\[
\frac 12\sum_i|x_i|\leq |Z|
\]
where $Z$ is the initial mismatch vector,
we obtain
\[
|\hat{C}_0+\hat{R}_1|\leq \frac{2}{\kappa}|Z|.
\]
Writing $|Z|\leq 4|\error|$ from Lemma~\ref{lem:preproc}, we get 
\[
|\hat{C}_0+\hat{R}_1|\leq \frac {8}{\kappa}|\error|.
\]
Since $\kappa <1$, we may write $2|\error|\leq \frac{2}{\kappa}|\error|$, and
now we have
\[
|\error + \hat{\error}|\leq \frac{2}{\kappa}|\error|+|\hat{C}_0 +
\hat{R}_1|\leq\frac{10}{\kappa}|\error|
\]
which is smaller than the minimum distance of the quantum code whenever
$|\error|\leq \frac{\kappa^3\delta^2}{2^{12}} \frac{n}{\Delta} < \frac{\kappa^3\delta^2}{10.2^{8}} \frac{n}{\Delta}$.
\end{proof}

\color{black}

\section*{Acknowledgements}
We thank Ting-Chun Lin for bringing up an issue with the parallelisation of the decoder in a previous version of this manuscript.
We acknowledge support from the Plan France 2030 through the project ANR-22-PETQ-0006. GZ also acknowledges support from the ANR through the project QUDATA, ANR-18-CE47-0010.

%


\end{document}